\documentclass[]{article}
\usepackage{fullpage}
\usepackage{graphicx}
\usepackage[outdir="./"]{epstopdf}
\usepackage{amsmath}
\usepackage{amssymb}
\usepackage{amsthm}
\usepackage{xcolor}
\usepackage[ruled,vlined]{algorithm2e}
\usepackage{todonotes}
\usepackage[noend]{algcompatible}
\usepackage{float}
\usepackage{svg}
\usepackage{caption}

\newtheorem{theorem}{Theorem}
\newtheorem{definition}{Definition}
\newtheorem{lemma}[theorem]{Lemma}
\newtheorem{corollary}{Corollary}[theorem]
\newtheorem{manuallemmainner}{Lemma}
\newenvironment{manuallemma}[1]{%
  \renewcommand\themanuallemmainner{#1}%
  \manuallemmainner
}{\endmanuallemmainner}

\algnewcommand\algorithmicinput{\textbf{INPUT:}}
\algnewcommand\INPUT{\item[\algorithmicinput]}

\newtheorem*{claim}{Claim}

\title{Musings on the HashGraph Protocol:\\ Its Security and Its Limitations}
\author{
Vinesh Sridhar\\
\texttt{vsridhar@umd.edu}
\and
Erica Blum\\
\texttt{erblum@umd.edu}
\and
Jonathan Katz\\
\texttt{jkatz@cs.umd.edu}
}

\begin{document}

\maketitle
\begin{abstract}
    The HashGraph Protocol is a Byzantine fault tolerant atomic broadcast protocol. Its novel use of locally stored metadata allows parties to recover a consistent ordering of their log just by examining their local data, removing the need for a voting protocol. Our paper's first contribution is to present a rewritten proof of security for the HashGraph Protocol that follows the consistency and liveness paradigm used in the atomic broadcast literature. In our second contribution, we show a novel adversarial strategy that stalls the protocol from committing data to the log for an expected exponential number of rounds. This proves tight the exponential upper bound conjectured in the original paper. We believe that our proof of security will make it easier to compare HashGraph with other atomic broadcast protocols and to incorporate its ideas into new constructions. We also believe that our attack might inspire more research into similar attacks for other DAG-based atomic broadcast protocols.
\end{abstract}
\section{Introduction}\label{sec:intro}

Say that you have several databases around the world and want them all to store identical copies of data that you input over time. This level of redundancy is common in cloud storage applications for example because it makes the services resilient to database failures and errors accumulated during data transmission. The most obvious solution has us construct a centralized system in which a single server disseminates the data to the rest. However, this produces a single point of failure and requires recipients to trust the central server. These drawbacks motivate the study of decentralized systems, in which trust and resiliency are spread over many parties rather than just one. 

Blockchain protocols, in which parties maintain a distributed log of transactions, are a popular application of this concept. For such protocols to be practical, parties must be able to maintain consistent local logs and commit new transactions to their log promptly. When a protocol is able to do this, we say that it solves the \emph{atomic broadcast} problem. More formally, we require any atomic broadcast protocol to have these two properties: \emph{consistency}, that everyone's logs agree, and \emph{liveness}, that an inputted transaction will eventually be placed in every party's log. We would like this to hold even if transactions may be delayed and reordered and a fraction of the parties act arbitrarily.

The HashGraph Consensus Protocol \cite{swirlds} \cite{swirlds2} is an example of an atomic broadcast protocol. Often these protocols are divided into two stages: Parties first distribute their transactions and thereafter vote on whose transactions will be added to the network and in what order. HashGraph differentiates itself from other constructions, such as HoneyBadgerBFT \cite{honeybadgerbft}, Dumbo \cite{dumbo}, and BEAT \cite{beat}, in the way it handles voting on the final ordering of the log. Rather than utilizing a separate voting protocol, votes are embedded into the structure of each party's log itself. The authors call the data structure that represents the log a \emph{hashgraph}. It is a directed acyclic graph (DAG) that not only stores the transactions a party has received, but also the communication history that led up to the distribution of that transaction. As it accumulates transactions, a party's own hashgraph eventually contains enough information to let them decide on a globally consistent transaction log. As long as parties can maintain consistent hashgraphs, they can eventually output transactions in a consistent order.

Our first contribution is a more formal analysis of HashGraph's consistency and liveness properties than the ones in the original papers \cite{swirlds} \cite{swirlds2}, showing that it is indeed a secure atomic broadcast protocol. On the negative side, we show that although liveness holds it may take $\Theta(2^n)$ rounds for a transaction to be committed to a party's log. The original HashGraph papers \cite{swirlds} \cite{swirlds2} only showed an $O(2^n)$ upper bound; we have proved it tight.

In Section~\ref{sec:related-work}, we discuss related work. In Section~\ref{sec:model}, we explain the model and define relevant terms. Section~\ref{sec:protocol} describes the HashGraph data structure and protocol. In Section~\ref{sec:proof}, we present our consistency and liveness proofs, and in Section~\ref{sec:delay} we show that transactions may require a long time to be committed. Lastly, in Section~\ref{sec:discussion} we provide some concluding remarks.
\section{Related Work}\label{sec:related-work}
The original HashGraph Consensus paper was published in 2016 \cite{swirlds}, and a sequel that clarified the protocol was published in 2020 \cite{swirlds2}. Lasy \cite{lasy} presents an analysis of the protocol and several variants, focusing on empirical performance tests. Crary~\cite{verify-hashgraph} also presents a proof of security for the HashGraph protocol, but their article mainly focuses on applying the original paper's proof to the Coq proof assistant. We aim to rewrite the proof to emphasize how the protocol satisfies liveness and consistency, making the proof more in line with other atomic broadcast protocol proofs of correctness.

In the HashGraph protocol, parties maintain a directed acyclic graph called a hashgraph to store metadata about messages sent and received. Other DAG-based asynchronous Byzantine fault tolerant consensus protocols include Aleph~\cite{aleph}, DAG-Rider~\cite{dag-rider}, Bullshark~\cite{giridharan2022bullshark}, Tusk~\cite{narwhal-tusk}, and JointGraph~\cite{jointgraph}.
\section{Model}\label{sec:model}

Let $n$ denote how many parties participate in the protocol. Let $t$ be an upper bound on the number of \emph{corrupted} parties. These parties may act arbitrarily, and we may consider them coordinated by a single entity called the \emph{adversary}. The remaining parties are \emph{honest} and follow the protocol exactly as specified. Throughout, we assume that $t < n/3$. We consider a \emph{static} adversary that chooses which parties to corrupt at the beginning of execution. 

We also consider an asynchronous network, where message delays are unbounded and the adversary may delay and reorder messages arbitrarily. However, we assume that messages must be delivered eventually. 

The HashGraph protocol solves the \emph{atomic broadcast} problem, which we will now formally describe.

\begin{definition}[Atomic Broadcast] 
Let $\Pi$ be a protocol executed by $n$ parties $p_1, \dots, p_n$, in which each party receives transactions from an external mechanism and outputs to a write-once log of transactions. We say that $\Pi$ is a \emph{secure atomic broadcast protocol} if it has the following properties:
\begin{itemize}
	\item (Consistency) Let $log_i, log_j$ be logs held by honest $p_i, p_j$ (possibly at different points in time). Then, $log_i$ is a prefix of $log_j$ or vice-versa.
	\item (Liveness) Every transaction placed in an honest party's buffer is eventually included in every honest party's log.
\end{itemize}
\end{definition}

%Note that in our definition of \emph{consistency}, $p_i$ and $p_j$ may refer to the same party at two different points in time. Trivially, a protocol could satisfy consistency by having parties maintain empty, and so equivalent, logs. Thus, we also require liveness.

The HashGraph protocol also assumes that the parties have established a public key infrastructure (PKI), allowing them to sign arbitrary messages. We assume the signature scheme is unforgeable, preventing the adversary from spoofing a message to seem as if it were sent by an honest party. Finally, the protocol assumes that all parties agree on some collision-resistant hash function.
\section{Consensus Protocol}\label{sec:protocol}

Throughout the protocol, each party builds and processes their own \emph{hashgraph}. A hashgraph is a directed acyclic graph which represents the flow of gossip through the network from its owner's perspective. Its vertices store the messages it has received from other parties and its edges help determine the set of parties that propagated that message before it was first seen by the hashgraph's owner. Thus, when a party $p$ learns a new message, it also learns the chain of parties that gossiped the message before $p$ received it. We will see below that enough information is embedded in a hashgraph to allow parties to commit transactions to their log without needing a networked voting protocol.

At a high level, the protocol works like so: Upon receiving a message $m$, a party adds it as a vertex to its hashgraph. It then processes its locally-stored hashgraph to see if any new information should be committed to the log. Then it produces a message of its own that it adds to its hashgraph. This message points to $m$, encoding the chain of gossip, and contains a set of transactions the party wants to add to everyone's log. The networked and local aspects of the protocol trigger each other in this way indefinitely, as the parties construct an ever-increasing log of transactions.

We will first describe the hashgraph data structure and the way it can be used to commit transactions in more detail. Then we will discuss the protocol itself and how it applies the properties of a hashgraph to implement atomic broadcast.

\subsection{The HashGraph data structure}
The hashgraph data structure is a directed acyclic graph. We call the messages sent during the protocol \emph{events}. The vertices in a hashgraph represent events that its owner has either created or validated upon receipt. As noted above, honest parties create events upon receiving an event from another party. An event contains a set of transactions that its creator wants to commit to the log as well as these pieces of metadata: a creator ID, a timestamp from the creator's local clock, and a cryptographic signature by the creator. It also holds hashes of two other events, called a \emph{self-parent} and an \emph{other-parent} respectively. These effectively act as parent pointers, and we represent them as such in hashgraph diagrams. If an event was created by an honest party, we call it an \emph{honest event}. 

Given some event $y$, we will use the notation $y.transactions$, $y.timestamp$, $y.creator$, etc. to refer to these fields. We will now formally define the event properties mentioned above.

\begin{definition}[self-parent]
An honest event's {\bf self-parent} is its creator's previous event. 
\end{definition}

\begin{definition} [other-parent]
An honest event's {\bf other-parent} is some earlier event created by a different party. Honest parties only create events upon receiving an event from another party.
\end{definition}

\begin{definition} [ancestors and descendants]
An {\bf ancestor} of some event $x$ is any event that can be reached by traversing parent pointers in the subgraph rooted at $x$. An ancestor of $x$ is a {\bf self-ancestor} if it can be reached solely using self-parent pointers. 
If some event $y$ is an ancestor of $x$, we say $x$ is a {\bf descendant} of $y$.
\end{definition}

The above description applies to all honest events except the very first ``genesis" event each party creates. A genesis event is a dummy event with no parents that parties create on their own in order to start the protocol. 

Figure~\ref{fig:example} is a diagram of a hashgraph. The vertices represent events, and the edges represent parent relationships. Events in the same column are created by the same party and time flows upward. As a result, when an honest party creates an event, that event's self-parent is directly beneath it and its other-parent lies in a different column. In the diagram, we say that event $A$ is an ancestor of event $D$. Likewise, event $D$ is a descendant of event $A$.

\begin{figure}[ht]
    \centering
    %\includesvg[scale=.7]{images/Example.svg}
    \includegraphics[scale=0.7]{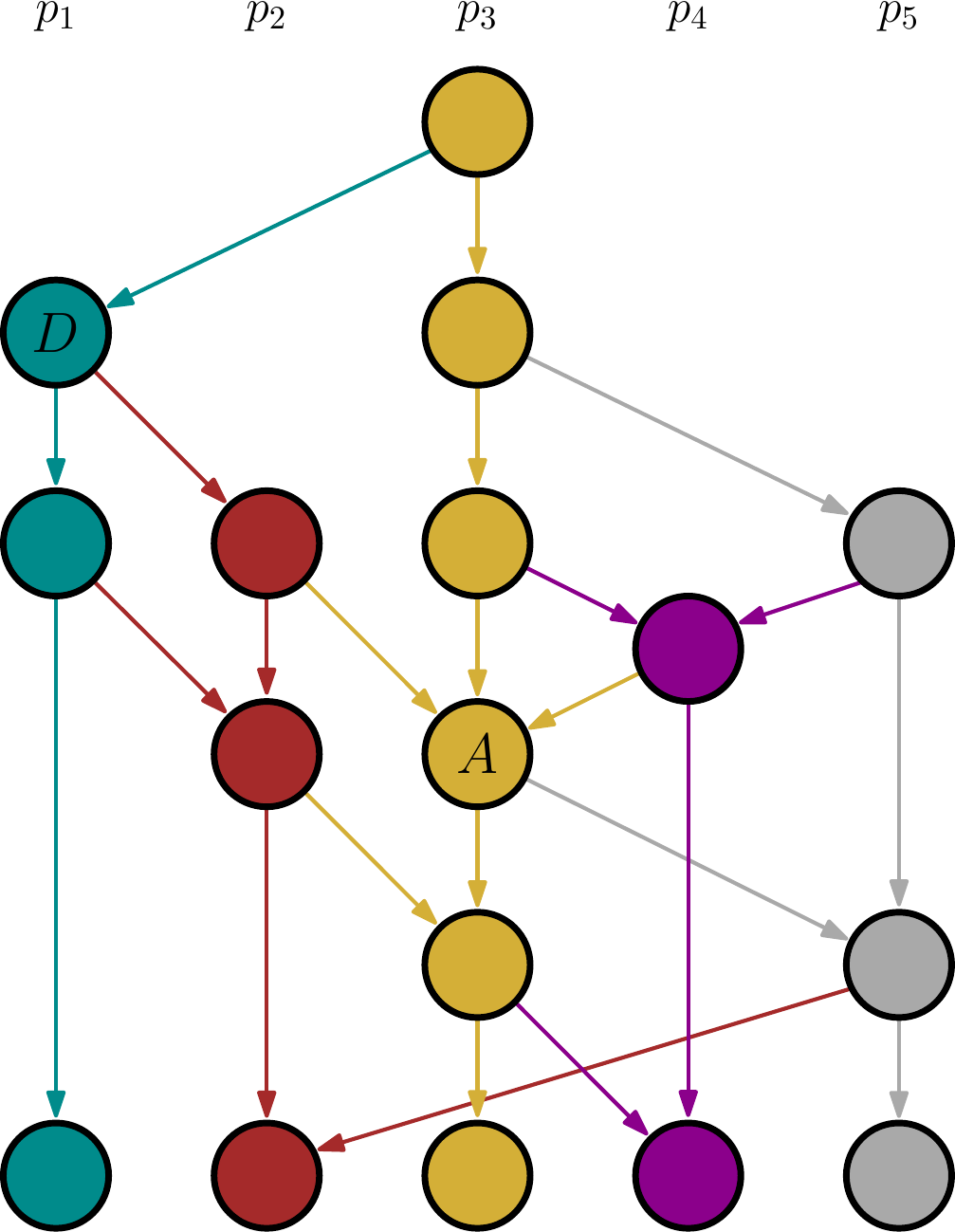}
    \caption{An example hashgraph where $n = 5$.}
    \label{fig:example}
\end{figure}

A hashgraph is supposed to reveal what participating parties have learned over time. For example, if party $p_1$ holds the hashgraph in Figure~\ref{fig:example}, then columns $2$ through $5$ tell $p_1$ about what events $p_2$ through $p_5$ are aware of. To maintain this property, we require that an honest party only adds a new event $y$ to its hashgraph after receiving all of $y$'s ancestors. The ancestors of $y$ indicate the events $y.creator$ knew about and gossiped about at the time it created $y$, so if we were to omit this requirement and consider $y$ in isolation, the hashgraph would no longer track the spread of gossip through the network. 

Because honest events are only created upon receiving another event and $p_1$ only adds some event $y$ to its hashgraph after receiving all of $y$'s ancestors, $p_1$ can learn every event that $y$'s creator knows by tracing through $y$'s ancestors in its own hashgraph. This would be sufficient if all parties were honest, but the adversary might try to disrupt this system by introducing \emph{forks} into honest hashgraphs. 

\begin{definition}[fork]
	A {\bf fork} is a pair of events $(x, y)$ from the same creator such that $x$ is not an ancestor of $y$ and $y$ is not an ancestor of $x$. 
\end{definition}

A corrupted party may choose to fork by creating two events $(z, z')$ that point to the same self-parent. $z$ and $z'$ may have different other-parents, hold different transactions, have conflicting timestamps and signatures, or a combination of the above. Ultimately, this allows the adversary to obscure the events it really knows by cultivating two disjoint branches in its hashgraph. It can then tell some honest parties it has some set of events and other honest parties that it has some other set of events. This temporarily causes their hashgraphs to differ. Even once all honest parties discover the fork, they then must agree upon a branch to process and a branch to ignore. Figures~\ref{fig:fork-create}~and~\ref{fig:fork-resolve} show an example of a fork. We will see how the protocol mitigates the effect of forks below.

\begin{figure}
\begin{center}
\begin{minipage}[t]{0.45\textwidth}
\centering
  %\includesvg[scale=0.7]{images/fork 2.svg}
  \includegraphics[scale=0.7]{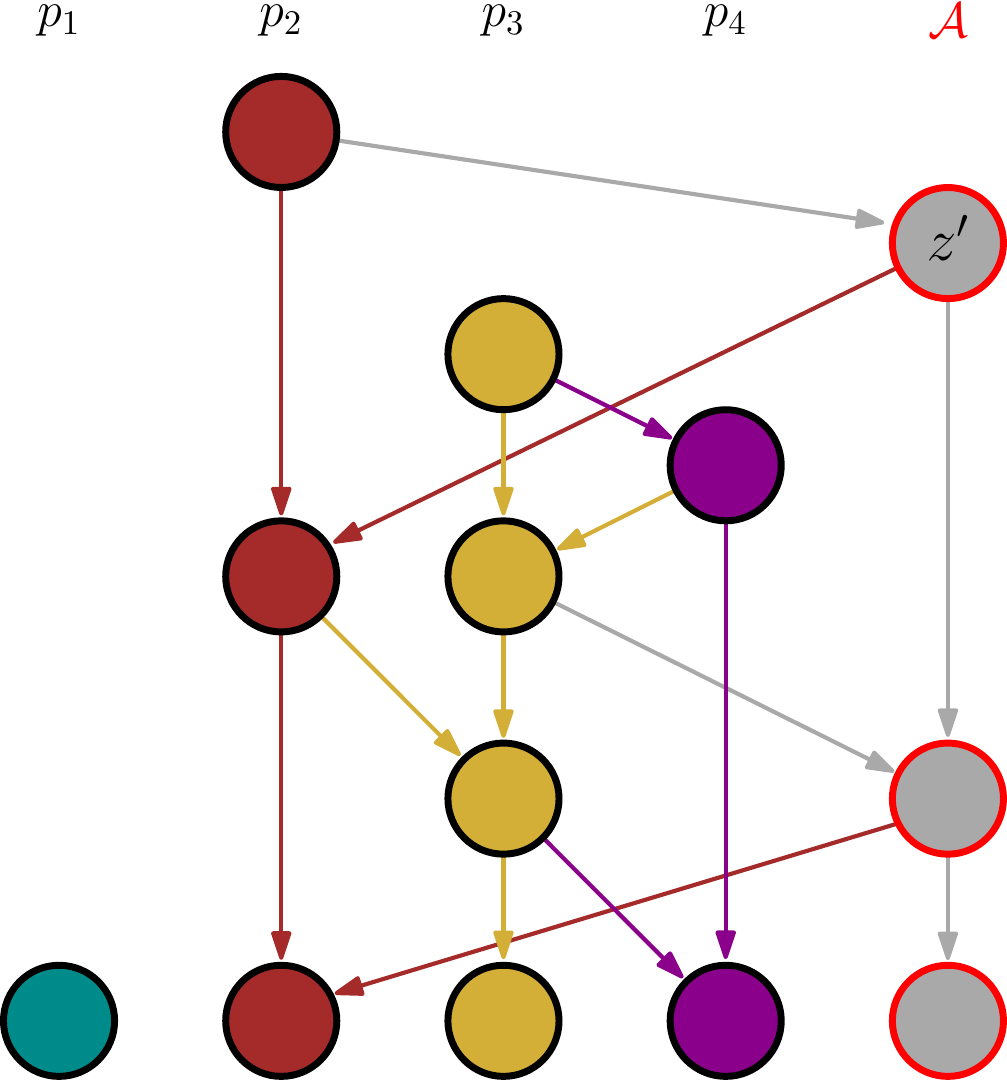}
\end{minipage}\hfill
\begin{minipage}[t]{0.45\textwidth}
\centering
  %\includesvg[scale=0.7]{images/fork.svg}
  \includegraphics[scale=0.7]{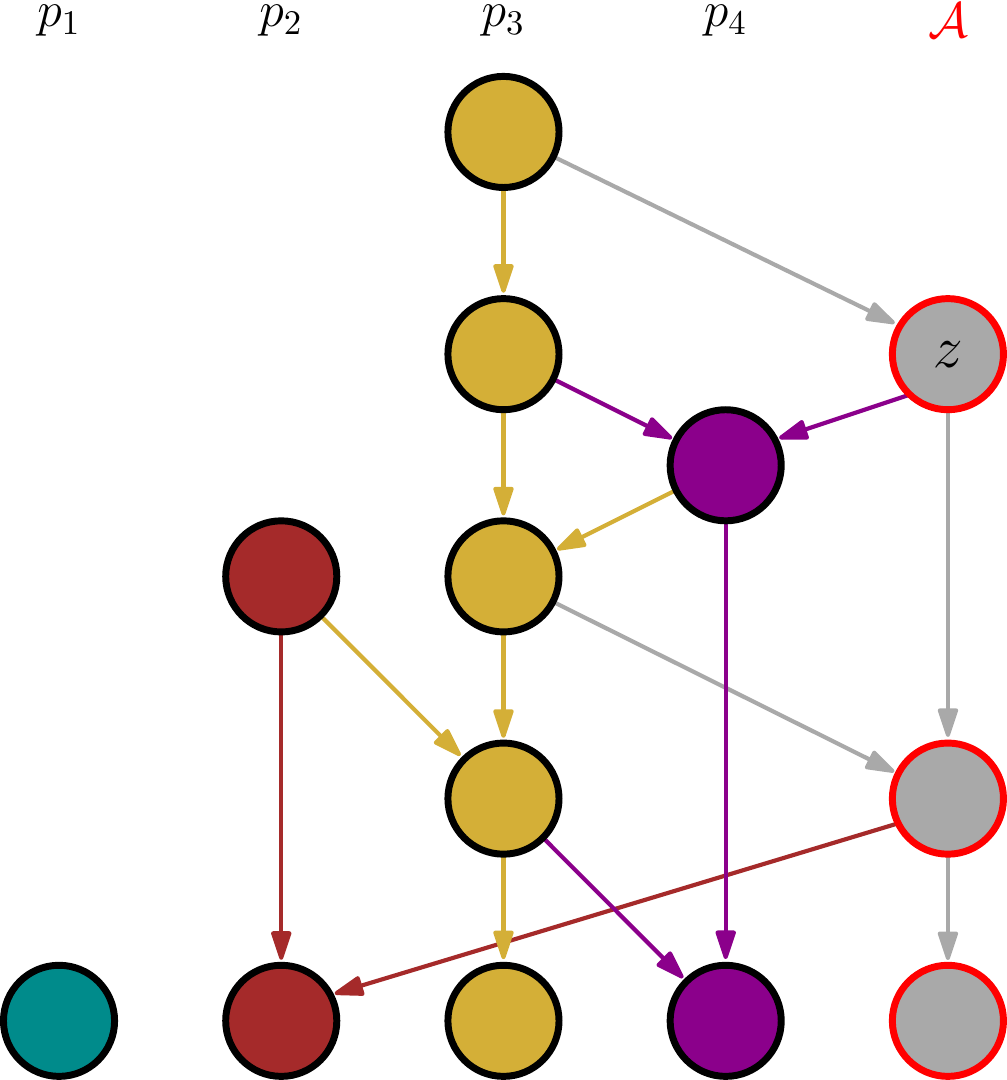}
\end{minipage}\hfill
\caption{Honest party $p_2$'s hashgraph is on the left and honest party $p_3$'s hashgraph is on the right. Vertices in the same locations represent the same events held by both parties. The adversary has created a \emph{fork} by gossiping the conflicting events $z$ and $z'$.}\label{fig:fork-create}
\end{center}
\end{figure}

\begin{figure}
\centering
  %\includesvg[scale=0.7]{images/fork 3.svg}
  \includegraphics[scale=0.7]{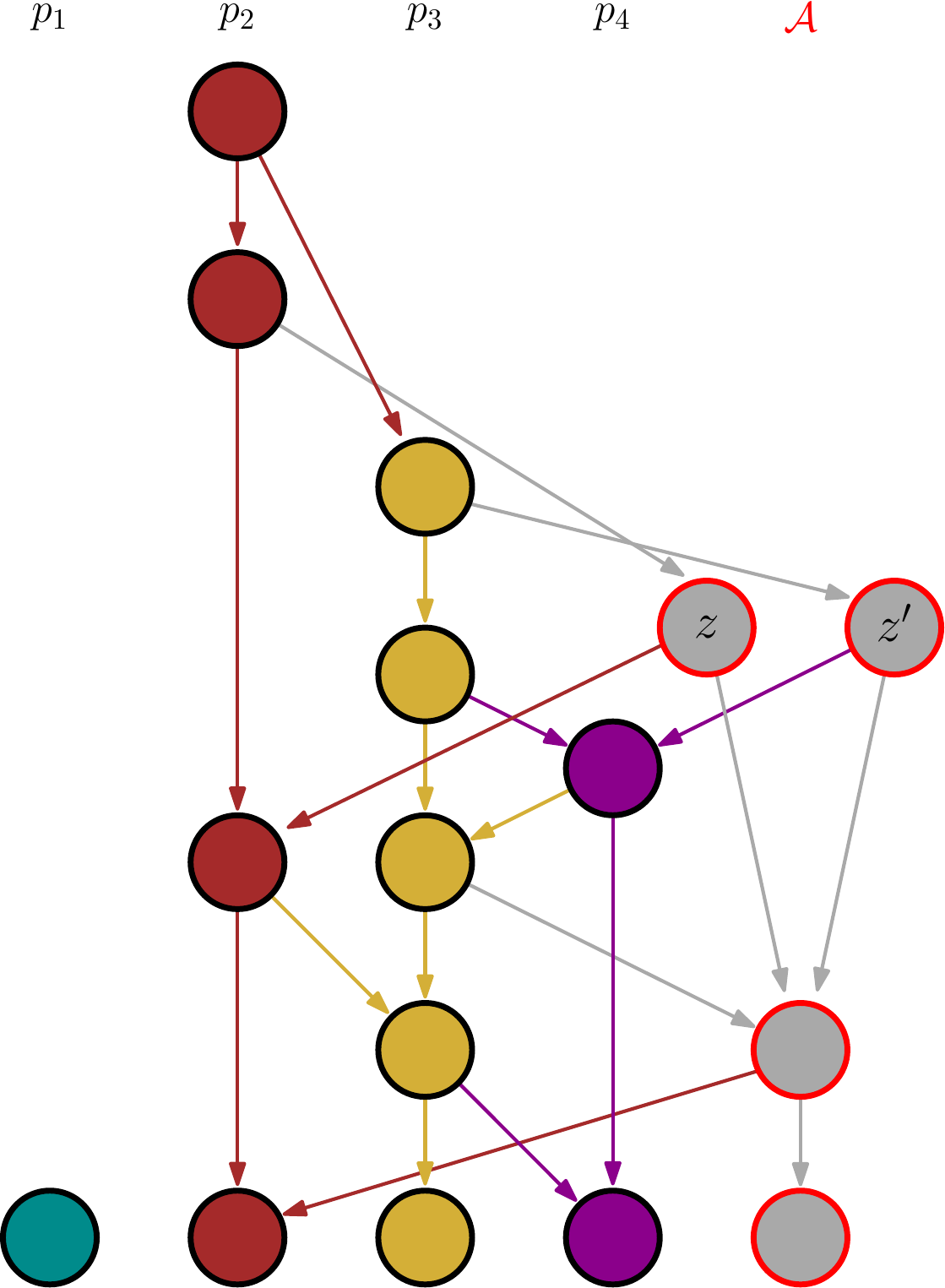}
\caption{$p_2$'s hashgraph upon receiving both sides of the fork. Note that $x$ now \emph{observes} the fork $(z, z')$ and so \emph{sees} neither $z$ nor $z'$.}\label{fig:fork-resolve}

\end{figure}

We say an event $x$ \emph{observes a fork} if $x$ is a descendant of two events $z$ and $z'$ such that $(z, z')$ is a fork.
Honest parties never create forks, so a fork betrays that its creator is corrupted. Similarly, other relationships between events disclose information about other parties. 
Below we define \emph{seeing}, \emph{strongly seeing}, and \emph{graph-consistency}.

\begin{definition}[seeing]
	Fix some hashgraph $H$. An event $x$ {\bf sees} an event $y$ in $H$ if $y$ is an ancestor of $x$ in $H$ and $x$ does not observe a fork from $y$'s creator.
\end{definition}

\begin{definition}[strongly seeing]
	Fix some hashgraph $H$. An event $x$ {\bf strongly sees} an event $y$ in $H$ if $x$ sees a set of greater than $2n/3$ events with distinct creators that each see~$y$.
\end{definition}

\begin{definition}[graph-consistent]
	Let $A$ and $B$ be two hashgraphs. $A$ and $B$ are {\bf graph-consistent} if for every event $x$ that appears in both $A$ and $B$, the subgraph containing $x$'s ancestors is the same in $A$ and $B$.
\end{definition}

Note that the original paper uses the term \emph{consistent} to refer to this property \cite{swirlds}. We have opted to use the term \emph{graph-consistent} to reduce confusion between this and consistency for atomic broadcast.

\subsection{The HashGraph Protocol}

\begin{algorithm}
\caption{The HashGraph Protocol from honest party $p_i$'s perspective. Adapted from \protect{\cite{swirlds}}}\label{algo:protocol-flow}
\begin{algorithmic}[1] 
\STATE // Initialization.
\STATE create genesis event $x_i$
\STATE
\STATE // Runs forever.
\LOOP 
\STATE \textbf{sync} $p_i$'s hashgraph with a randomly chosen party.
\ENDLOOP
\STATE
\STATE \textbf{upon} receiving some new event $y$:
    \STATE \hspace{.5cm} \textbf{if} $p_i$ has all of $y$'s ancestors \textbf{then}
    \STATE \hspace{1cm} add $y$ to $p_i$'s hashgraph
    \STATE \hspace{1cm} create a new event $z$ such that $z.self$-$parent \leftarrow$ $p_i$'s previous event and $z.other$-$parent \leftarrow y$
    \STATE \hspace{1cm} // Process updated hashgraph.
    \STATE \hspace{1cm} \textbf{call} \emph{divideRounds}
    \STATE \hspace{1cm} \textbf{call} \emph{decideFame}
    \STATE \hspace{1cm} \textbf{call} \emph{findOrder}
    \STATE \hspace{.5cm} \textbf{else}
    \STATE \hspace{1cm} buffer $y$ until $p_i$'s hashgraph holds all of $y$'s ancestors
\end{algorithmic}
\end{algorithm}

We will now describe the protocol in more detail (Pseudocode appears in Algorithm~\ref{algo:protocol-flow}). At the beginning of the protocol, each party multicasts their own genesis event. Unlike a regular event, this one has no parents. Parties then perpetually wait for new events. Consider some honest party $p$. Upon receiving some new event $y$ such that $p$'s hashgraph contains all of $y$'s ancestors, $p$ adds $y$ to its hashgraph. It then creates a new event whose self-parent is $p$'s previous event and whose other parent is $y$. Lastly, $p$ processes its hashgraph locally to see if it contains enough information to commit more events to its log. It does this by invoking three functions: \emph{divideRounds}, \emph{decideFame}, and \emph{findOrder}. 

In practice, parties share multiple events at the same time. This is called a \emph{sync}. When honest party $p_i$ syncs with honest party $p_j$, $p_j$ receives all of the events $p_i$ has. Even if it adds several events to its hashgraph, $p_j$ only creates a single event in response. That event's self-parent is $p_j$'s previous event and its other-parent is the latest event $p_i$ shared in its sync. By definition, this must be the event $p_i$ created right before it synced with $p_j$.

Incorrectly formatted events are discarded, and an honest party buffers a received event until it has added all of its ancestors to its hashgraph.

We will now describe how \emph{divideRounds}, \emph{decideFame}, and \emph{findOrder} allow a party to commit transactions just by examining its own hashgraph. We reproduce the original HashGraph paper's \cite{swirlds} pseudocode for these three functions in Appendix~\ref{sec:appendix}.

\subsubsection{divideRounds}
In the HashGraph Protocol, as in many other asynchronous protocols, round numbers express the progression of the protocol rather than being defined by a synchronized global clock. The purpose of \emph{divideRounds} is to assign each event in a hashgraph a round number. Genesis events are assigned round 1. In general terms, if an event in some honest party $p$'s hashgraph is assigned round $(r+1)$, a sufficient number of parties have contributed round-$r$ events to $p$'s hashgraph. Rounds will not progress until the adversary allows sufficiently many parties to sync with each other. This limits how the adversary can hide events from some honest parties, in particular constraining its ability to keep forks a secret.

We will now describe the procedure more formally, starting by defining the term \emph{witness}.
\begin{definition}[witness]
    An event $x$ is a {\bf witness} in round $r$ if its self-parent is in some round $r' < r$.
\end{definition}

In the procedure we iterate through every event. Fix some event $x$ in some honest hashgraph $H$. Let $p$ be $H$'s owner. We first let $r \leftarrow max\{x.self\text{-}parent.round, x.other\text{-}parent.round\}$. If $x$ strongly sees greater than $2n/3$ round-$r$ witnesses from distinct parties, it is promoted to round $(r+1)$. Otherwise, it remains in round $r$. This enforces the property we described above: for there to be greater than $2n/3$ round-$r$ witnesses from distinct parties, greater than $2n/3$ parties must have received and sent events in round $r$. The round system forces the adversary to deliver messages from at least $2n/3$ parties in order to progress the protocol, exposing forks and ensuring that a large portion of honest messages are delivered.

To compute the round of an event, we only need to examine that event's ancestors. In addition, $p$ only adds an event to $H$ if it has all of that event's ancestors. Therefore, an honest party can immediately compute an event's round upon adding it to its hashgraph, and once an event is assigned a round, that assignment will never change.

\subsubsection{decideFame}
After using \emph{divideRounds} to identify witnesses, we invoke \emph{decideFame} to subdivide them into those that are \emph{famous} and those that are \emph{not famous}. 

\begin{definition}[famous witness]
    An event $x$ is a {\bf famous witness} if it is a witness decided famous by the \emph{decideFame} procedure. This procedure examines witnesses in future rounds and allows them to vote on the famousness of $x$. Famousness roughly means that the witness was distributed to many parties quickly. A famous witness is {\bf unique} if it is the only famous witness produced by its creator that round.
\end{definition}

Famous witnesses give us a set of events that are quickly received by many parties. They help enforce agreement upon which events will be committed and when. We will see this in more detail in our explanation of \emph{findOrder}. Here we will explain how witnesses are elected famous. 

Fix some honest hashgraph $H$, its owner $p$, and some round-$r$ witness $x$ in $H$. To decide $x$'s famousness, the \emph{decideFame} procedure iterates through witnesses in each subsequent round. Each witness is assigned a vote. There are two cases for assigning votes: a round-$(r+1)$ witness votes ``yes" on $x$'s famousness if it sees $x$ in $H$ and ``no" otherwise. In some later round $r' > r+1$, a round-$r'$ witness votes by taking the majority vote of the round-$(r'-1)$ witnesses it strongly sees. If any witness sees more than $2n/3$ votes for some $v$, $x$'s fame is determined to be $v$.

Periodically, set by the parameter $c$, a \emph{coin round} occurs. Call this round $r_c$. In this round, witnesses vote like so: if they strongly see more than $2n/3$ round-$(r_c-1)$ witnesses that voted for some $v$, then they take $v$ as their vote (as opposed to simply taking the majority vote). Otherwise, they vote by ``coin flip": they take the middle bit of their signature and vote ``yes" if it is a 1 and ``no" if it is a 0. 

The process described in the previous two paragraphs is repeated for each witness in $H$.

If $H$ has not accumulated enough events, we may run out of witnesses before we find one that observes a supermajority vote on $x$'s famousness. At that point, the party would declare $x$ ``undecided" and try again once it has received a new event. Interestingly, once an party has declared a witness famous or not famous, it will never change that outcome. This is because the way a witness votes is based either on its signature, which is fixed, or on its ancestors. The latter is similar to what we saw in \emph{divideRounds}, so witnesses will always vote the same way in a witness election for some fixed $x$. Thus, once enough events have accumulated such that some witness decides $x$'s famousness, that witness will always decide $x$'s famousness in the same way regardless of how many more events are added to $H$. 

\subsubsection{findOrder}\label{find-order}
\emph{divideRounds} allows us to define witnesses and \emph{decideFame} allows us to find witnesses that are quickly disseminated to many parties. \emph{findOrder} uses the resulting famous witnesses to determine when there is enough information shared amongst honest parties to commit an event to their logs.

The goal of the procedure is to assign three metrics to each event and use them to compute a total ordering over events, the second and third metric breaking any ties that arise. It works like so: fix some honest hashgraph $H$ and its owner $p$. We iterate through each event $x$ that $p$ has not yet committed to its log. The first metric we assign to $x$ is called a \emph{roundReceived}. 

\begin{definition}[roundReceived]
    An event $x$'s {\bf roundReceived} is the first round $R_x$ in which all unique famous witnesses see $x$ and the fame of all witnesses in all rounds $r \leq R_x$ has been decided. 
\end{definition}

Note that $x$'s roundReceived must be greater than or equal to $x.round$. If $H$ currently does not have a round in which all famous witnesses see $x$, then $x$'s roundReceived remains undecided. However, once $x$ has been assigned a roundRecieved, it will never change because $p$ cannot retroactively change the structure of its hashgraph. We also show in Theorem~\ref{liveness} of Section~\ref{sec:proof} that $p$ will eventually receive enough events such that $x$ is assigned a roundReceived. 

The next metric is called $x$'s \emph{logical timestamp}.

\begin{definition}[logical timestamp]
    An event $x$'s {\bf logical timestamp} is found after it is assigned a roundReceived $R_x$. To compute it we initialize an empty set $S$. For each round-$R_x$ unique famous witness $y$, add the self-ancestor of $y$ that is the closest descendant of $x$ to $S$. $x$'s logical timestamp is the median timestamp of $S$.
\end{definition}

The reason for using a logical timestamp instead of $x.timestamp$ is that if $x$ were created by a corrupted party, that party could have assigned $x.timestamp$ arbitrarily. The logical timestamp sidesteps this issue by having many different events contribute to the metric. 

Call $x$'s roundReceived $R_x$. Because the events that appear in $S$ are ancestors of the round-$R_x$ famous witnesses, $H$ must contain them all at the time $x$ is assigned its roundReceived. Thus, $p$ can compute $x$'s logical timestamp upon assigning $x$ a roundReceived.

The last metric we use is $x$'s \emph{whitened signature}, computed by XORing $x$'s signature with the signatures of each round-$R_x$ unique famous witness. To commit $x$ to its log, $p$ places $x$ in the appropriate location sorted first by roundReceived, then by logical timestamps, and lastly by whitened signatures.  

In the next section, we will formally prove how these functions make the HashGraph protocol live and consistent.

\section{A New Proof of Security}\label{sec:proof}
In this section, we present an alternate proof of the security of the HashGraph protocol. Our goal is to present the main result of the HashGraph papers \cite{swirlds} \cite{swirlds2} in a way that is more in line with the atomic broadcast literature by providing separate proofs of how the protocol attains \emph{consistency} and \emph{liveness}. Our strategy to prove that the protocol is \emph{consistent} is to show that if two hashgraphs are graph-consistent, the local computations run on them produce identical results. Lemmas 1 through 6 culminate in Theorem 7, our proof of the HashGraph Protocol's consistency. We then show that the protocol has \emph{liveness} in Theorem 8, independent of the previous proofs. We rely on lemmas $5.11$, $5.12$, and $5.18$ proved in the original paper \cite{swirlds}. We have reproduced them below and have provided a brief intuition for each of their proofs.
\begin{manuallemma}{5.11}[\cite{swirlds}]
\label{lemma:graph-consistent}
All parties have graph-consistent hashgraphs.
\end{manuallemma}

This is trivially proven by the collision resistance of a hash function. The only way for two honest parties to not have graph-consistent hashgraphs would be if the adversary were able to compute two events that map to the same hash. This would allow it to manufacture two hashgraphs that appear identical in terms of their hashes, but are actually different. If the hash function is collision resistant, this will occur with negligible probability.

\begin{manuallemma}{5.12}[Strongly Seeing Lemma \cite{swirlds}]
\label{lemma:strongly-seeing}
If the pair of events $(x,y)$ is a fork, and $x$ is strongly seen by event $z$ in hashgraph $A$, then $y$ will not be strongly seen by any event in any hashgraph $B$ that is graph-consistent with $A$.
\end{manuallemma}

If, in $A$, $z$ strongly sees $x$ then $z$ sees events from at least $2n/3$ distinct parties that each see $x$. Let us assume that there is some $z'$ in $B$ that strongly sees $y$. Then $z'$ also sees events from at least $2n/3$ distinct parties that each see $y$. Because there can be no more than $n/3$ corrupted parties, there must be some honest party $p$ such that $z$ sees an event from $p$ that sees $x$ and $z'$ sees an event from $p$ that sees $y$. But if $p$ is honest, it cannot have events that see both sides of a fork. This contradicts our assumption. 

Lemma~\ref{lemma:late-witness} concerns \emph{late witnesses}. Consider two graph consistent hashgraphs $A$ and $B$ and some round-$r$ witness $x$. If $x$ is not held by $A$ and if some other round-$r$ witness $y$ has had its famousness decided in $A$, then we call $x$ a late witness. If $B$ contains $x$, then this lemma shows that $B$ will decide that $x$ is ``not famous". As we described in the previous section, famousness is a measure of how quickly a witness spreads to many parties. If $A$ has enough witnesses in later rounds to be able to decide $y$'s famousness but has not yet received $x$, then $x$ did not spread quickly enough to be famous. 

\begin{manuallemma} {5.18} [\cite{swirlds}]
\label{lemma:late-witness}
If hashgraph $A$ does not contain event $x$, but does contain all the parents of $x$, and hashgraph $B$ is the result of adding $x$ to $A$, and $x$ is a witness created in round $r$, and $A$ has at least one witness in round $r$ whose fame has been decided (as either famous or not famous), then $x$ will be decided as ``not famous" in $B$.
\end{manuallemma}

If $A$ has decided the famousness of a witness $y$ in round-$r$, then $A$ must contain some witness $w$ in some later round $r' > (r+1)$ that has decided $y$'s fame. Because $A$ does not contain $x$, we know that neither $w$ nor any of its ancestors see $x$. We also know that $B$ contains $w$ and its ancestors. In $x$'s famousness election in $B$, all of $w$'s round-$(r+1)$ ancestors will vote ``no" because they cannot see $x$. Now consider any ancestor $z$ of $w$ in round $(r+2)$ (possibly $w$ itself). It must be the case that $z$ strongly sees more than $2n/3$ round-$(r+1)$ witnesses and that none of those see $x$. Therefore, $z$ observes more than $2n/3$ ``no" votes. If $(r+2)$ is a normal round, $x$ is determined to be not famous. If $(r+2)$ is a coin round, note that any ancestor of $w$ that is a witness will not flip a coin because they observe more than $2n/3$ ``no" votes. As a result, a round-$(r+3)$ witness that is an ancestor of $w$ will determine that $x$ is not famous. 

\begin{lemma} [seeing is consistent]
\label{consistency:consistent-seeing}
For graph-consistent hashgraphs $A$ and $B$, consider two events $x$ and $y$ that are both in $A$ and $B$. If $x$ sees $y$ in $A$, then $x$ sees $y$ in $B$.
\end{lemma}
\begin{proof}
Assume towards a contradiction that $x$ sees $y$ in $A$ and $x$ does not see $y$ in $B$. 

By the definition of seeing, $y$ is an ancestor of $x$ in $A$. By the definition of graph-consistent, $y$ must be an ancestor of $x$ in $B$ as well. If $x$ does not see $y$ in $B$, $x$ must observe a fork by $y.creator$ in $B$. Furthermore, at least one side of this fork must be an ancestor of $y$. If $x$ sees $y$ in $A$, it must be true that $x$ does not observe a fork by $y.creator$ such that either side is an ancestor to $y$. However, if $A$ and $B$ are graph-consistent, then this cannot be true. This is a contradiction.
\end{proof}

\begin{figure}
\begin{center}
\begin{minipage}[t]{0.5\textwidth}
\centering
  %\includesvg[scale=0.7]{images/lemma 2 2.svg}
  \includegraphics[scale=0.7]{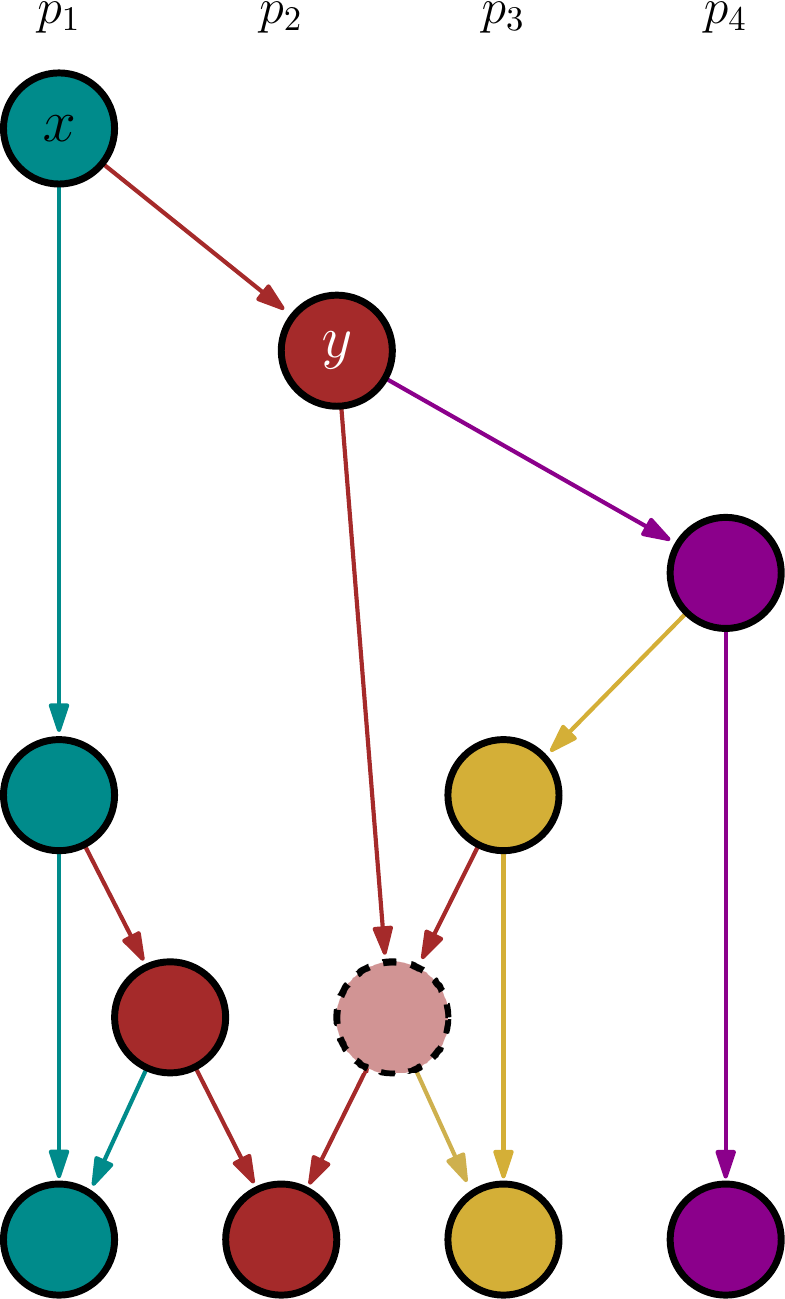}
\end{minipage}\hfill
\begin{minipage}[t]{0.5\textwidth}
\centering
  %\includesvg[scale=0.7]{images/lemma 2 1.svg}
  \includegraphics[scale=0.7]{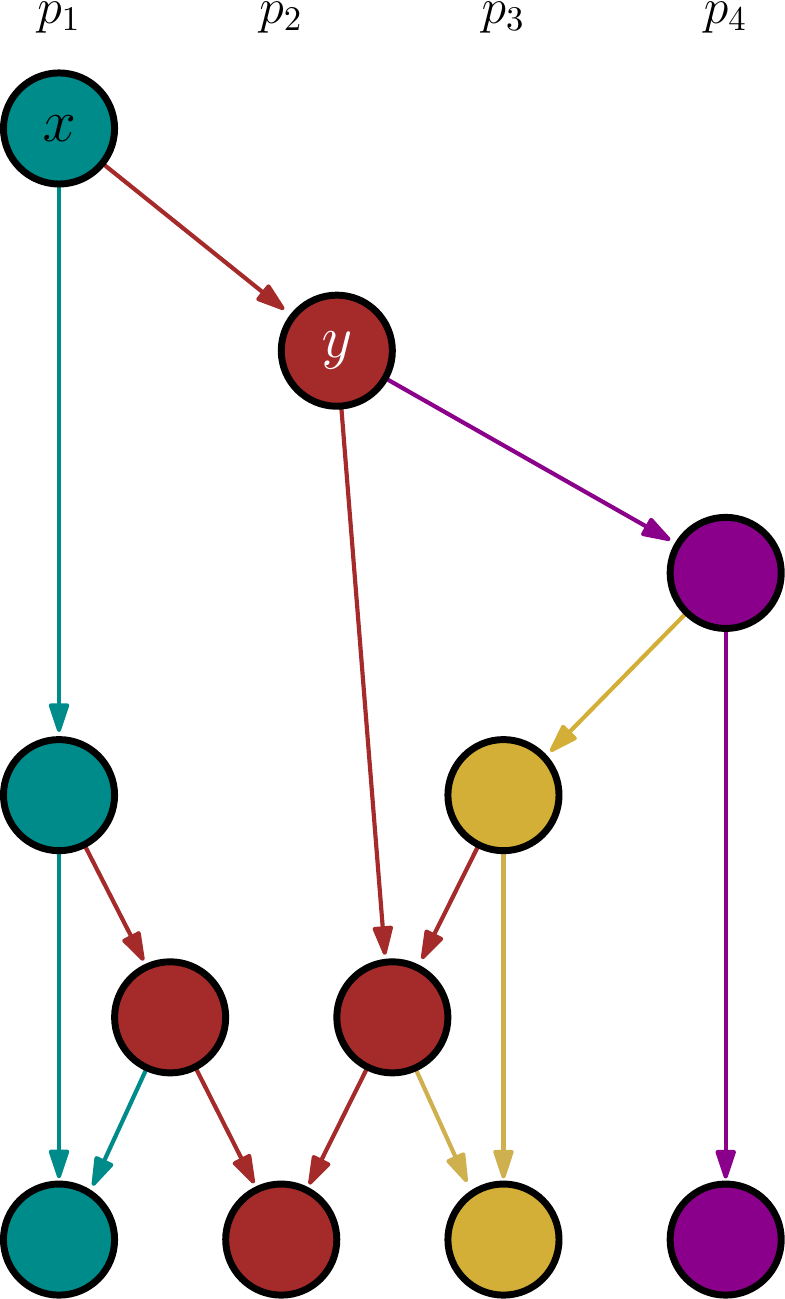}
\end{minipage}\hfill
\caption{In Lemma~\protect{\ref{consistency:consistent-seeing}}, we consider hashgraphs $A$ and $B$, on the left and right respectively. If $x$ observes a fork by $y.creator$ in $B$ but not in $A$, then the hashgraphs must not be graph-consistent.}\label{fig:lemma-fig}
\end{center}
\end{figure}

Figure~\ref{fig:lemma-fig} visualizes Lemma~\ref{consistency:consistent-seeing}. This lemma implies the following corollary.

\begin{corollary} [strongly seeing is consistent]
\label{consistency:consistent-str-seeing}
For graph-consistent hashgraphs $A$ and $B$, consider two events $x$ and $y$ that are both in $A$ and $B$. If $x$ strongly sees $y$ in $A$, then $x$ strongly sees $y$ in $B$.
\end{corollary}
\begin{proof}
If $x$ sees more than $2n/3$ events from distinct parties that see $y$, then $x$ strongly sees $y$. By graph-consistency and Lemma~\ref{consistency:consistent-seeing}, if $x$ sees some event $z$ in $A$, then $z$ is present in $B$ and $x$ sees it. Similarly, if $z$ sees $y$ in $A$, then $z$ sees $y$ in $B$ by Lemma~\ref{consistency:consistent-seeing} and graph-consistency with respect to $z$'s ancestors. Therefore, the conditions that determine whether $x$ strongly sees $y$ are the same between $A$ and $B$, so if $x$ strongly sees $y$ in $A$, then $x$ strongly sees $y$ in $B$.
\end{proof}

The following lemma formalizes the argument presented in Lemma 5.13 \cite{swirlds}.

\begin{lemma} [\emph{divideRounds} is consistent]
\label{consistency:divideRounds}
For graph-consistent hashgraphs $A$ and $B$, any event $x$ that is in both $A$ and $B$ is assigned the same round number.
\end{lemma}
\begin{proof}
We will prove the claim by induction.

In the base case, we consider some genesis event $x$ held by both $A$ and $B$. By definition, $x$ has no parents and is therefore assigned 1 as its round number.

For the inductive step, consider some event $x$ shared by $A$ and $B$. The two hashgraphs are graph-consistent, so $x$'s self- and other-parent must also be held by $A$ and $B$. Our inductive hypothesis is that in $A$ and $B$, $x$'s parents have matching round numbers. Let $r \leftarrow$ $max(x. self$-$parent.round $, $x.other$-$parent.round)$. By Corollary~\ref{consistency:consistent-str-seeing}, $x$ strongly sees the same set of round-$r$ witnesses from distinct parties in $A$ and $B$. If this shared set of strongly-seen witnesses is of size greater than $2n/3$, $x$ will be assigned round $(r+1)$ in $A$ and $B$. Otherwise, it will be assigned round $r$ in $A$ and $B$.
\end{proof}

\begin{lemma} [famousness votes are consistent]
\label{consistency:consistent-vote}
For graph-consistent hashgraphs $A$ and $B$, consider a witness $x$ voting on the famousness of some witness $y$. Both $x$ and $y$ are present in $A$ and $B$. If $x$ casts the vote $v_a$ in $A$ and $v_b$ in $B$ in $y$'s famousness election, then $v_a = v_b$.
\end{lemma}
\begin{proof}
We will prove the claim by induction.

In the base case, let $r \leftarrow y.round$ and let $x$ be some witness in round $(r+1)$. By Lemma~\ref{consistency:consistent-seeing}, seeing is consistent in graph-consistent hashgraphs. $x$'s vote is determined by whether it sees $y$, so its vote will be the same in both $A$ and $B$.

In the inductive step, let $x$ be a witness in some round $r' > r$. Assume that the votes of all round-$(r'-1)$ witnesses are consistent between $A$ and $B$.

By Corollary~\ref{consistency:consistent-str-seeing}, any round-$(r'-1)$ witness that $x$ strongly sees in $A$, $x$ must strongly see in $B$. Thus, in both hashgraphs, $x$ will strongly see the same set of round-$(r'-1)$ witnesses. As a result, if $r'$ is a normal round, $x$ will vote the same way in $A$ and $B$.

If round $r'$ is a coin round, then $x$ may flip a coin. If $x$ does not flip a coin, then the set of round-$(r'-1)$ witnesses must have more than $2n/3$ votes for $v$. If this is the case, then $v = v_a = v_b$ because $x$ strongly sees the same set of witnesses in both $A$ and $B$.

Otherwise, $x$ votes via coin flip. The coin flip returns the middle bit of $x$'s signature, which is the same in $A$ and $B$.

\end{proof}

\begin{lemma} [famousness is consistent]
For graph-consistent hashgraphs $A$ and $B$, if $A$ decides $v$ and $B$ decides $v'$ for the famousness of some witness $x$, then $v = v'$.
\label{consistency:consistent-fame}
\end{lemma}
\begin{proof} Assume towards a contradiction that $v \neq v'$.

Without loss of generality, assume that $A$ decides $x$'s famousness in the same round or before $B$.
Call the round in which $A$ decides $x$'s famousness $r$.

$B$ must contain more than $2n/3$ round-$(r-1)$ witnesses because if it has decided $x$'s famousness in some round $r' \geq r$, it must have an event in round $r$. By Lemma~\ref{lemma:strongly-seeing}, if a round-$r$ witness in $B$ strongly sees a round-$(r-1)$ witness $w_p$ created by party $p$ and a round-$r$ witness in $A$ strongly sees a round-$(r-1)$ witness $w'_p$ from $p$, then $w_p = w'_p$. Furthermore, Lemma~\ref{consistency:consistent-vote} implies that round-$r$ witnesses in $A$ and $B$ that strongly see $w_p$ and $w'_p$ respectively will receive the same vote from $w_p$ and $w'_p$. 

If $x$'s famousness was decided in round $r$ in $A$, some round-$r$ witness strongly saw more than $2n/3$ witnesses that voted for $v$ in round $(r-1)$. That means that fewer than $n / 3$ voted for $\neg v$. We have shown above that the round-$(r-1)$ witnesses vote consistently in $A$ and $B$. This implies that, in $B$, more than $2n/3$ round-$(r-1)$ witnesses vote for $v$. Therefore, every round-$r$ witness in $B$ observes a majority vote for $v$ in the set of round-$(r-1)$ witnesses it strongly sees, meaning that all round-$r$ witnesses in $B$ vote for $v$. If $B$ decides in round $r$, then $v = v'$.

If $B$ does not decide in round $r$ and round $(r+1)$ is a normal round, then a decision for $v$ will occur in round $(r+1)$, so $v'=v$.

If $B$ does not decide in round $r$ and round $(r+1)$ is a coin round, then every round-$(r+1)$ witness will strongly see more than $2n/3$ votes for $v$ from round-$r$ witnesses, meaning that none will flip coins and all will vote $v$. In round $(r+2)$, a decision for $v$ will occur, so $v'=v$. This contradicts our original assumption.
\end{proof}

\begin{lemma} [$roundReceived$ assignments are consistent]
Consider an event $x$ contained in two graph-consistent hashgraphs $A$ and $B$. If in $A$, $x.roundReceived \leftarrow r$ and in $B$, $x.roundReceived \leftarrow r'$, then $r = r'$.
\label{consistency:consistent-roundReceived}
\end{lemma}
\begin{proof}
Assume towards a contradiction that $r \neq r'$.

Without loss of generality, let $r < r'$. Then in hashgraph $A$, all round-$r$ unique famousness witnesses see $x$ and in hashgraph $B$, some round-$r$ unique famous witness $w$ does not see $x$. Therefore, $w$ must be famous in $B$ but not in $A$. Either $w$ is not present in $A$, or $A$ has $w$ but has decided that it is not famous. In the former case, when $A$ does receive $w$, it will determine that it is not famous by Lemma~\ref{lemma:late-witness}. Both cases violate Lemma~\ref{consistency:consistent-fame}, which is a contradiction.
\end{proof}

\begin{lemma} [logical timestamps are consistent]
For graph-consistent hashgraphs $A$ and $B$, upon assigning some event $x$ a roundReceived, both compute the same logical timestamp for $x$.
\label{consistency:consistent-timestamp}
\end{lemma}
\begin{proof}
Assume towards a contradiction that $A$ and $B$ compute different logical timestamps for $x$. By Lemma~\ref{consistency:consistent-roundReceived}, the two hashgraphs assign the same roundReceived value to $x$. Let $r \leftarrow x.roundReceived$.

In \emph{findOrder}, to produce a logical timestamp, a set of events $S$ is compiled. The events in $S$ are self-ancestors of the unique famous witnesses in round $r$ that are the closest descendants of $x$. Let $S_A$ be the set produced in $A$ and $S_B$ be the set produced in $B$. By the initial assumption, $S_A \neq S_B$.

Famousness is consistent by Lemma~\ref{consistency:consistent-fame}, so the set of round-$r$ famous witnesses is the same in $A$ and $B$. $S_A$ and $S_B$ are only populated by events that are self-ancestors of the round-$r$ unique famous witnesses, so they contain events from the same set of parties. 

If $S_A \neq S_B$ and both sets contain events from the same set of parties, there must be some party $p$ and some pair of events $y_A$ and $y_B$ such that $y_A \neq y_B$, $y_A.creator = y_B.creator$, $y_A \in S_A$, and $y_B \in S_B$. $(y_A, y_B)$ cannot be a fork, because then only one could be a self-ancestor of $y$'s round-$r$ unique famous witness. 

Without loss of generality, let $y_A$ be an ancestor of $y_B$. Because $A$ and $B$ are graph-consistent, if $B$ contains $y_B$ and $y_A$ is an ancestor of $y_B$, then $B$ must contain $y_A$. Then $y_A$ should be in $S_B$ instead of $y_B$, because it is a closer descendant to $x$ than $y_B$. Therefore, $S_A = S_B$, and so $A$ and $B$ compute the same logical timestamp for $x$, contradicting our original assumption.
\end{proof}

\begin{theorem}\label{consistency:theorem}
The HashGraph Protocol is consistent.
\end{theorem}
\begin{proof}
As we have shown in Lemmas~\ref{consistency:consistent-roundReceived}~and~\ref{consistency:consistent-timestamp}, each event $x$ received by graph-consistent hashgraphs $A$ and $B$ will be assigned the same $roundReceived$ and logical timestamp. These two properties are used to output events to a party's log. If any ties remain, they will be broken in a deterministic way by comparing adjacent events' signatures. Thus, graph-consistent hashgraphs will output $x$ to the same spot in their logs.

By Lemma~\ref{lemma:graph-consistent}, all honest parties maintain graph-consistent hashgraphs. Therefore, all honest parties will output $x$ to the same location in their logs.

We also require that events be committed in order. First, we will show that any two events $x$ and $y$ such that $x.roundReceived = y.roundReceived$ will be committed during the same call to \emph{findOrder}. Let $r = x.roundReceived$. By definition, all round-$r$ famous witnesses see $x$ and $y$. As a result, they must be descendants of $x$ and $y$. Consider the case in which an honest party $p$ commits $x$ before receiving $y$. If $y.roundReceived = r$, then $y$ must be an ancestor of all round-$r$ famous witnesses. An honest party waits to receive all ancestors of an event before placing that event in its hashgraph, so $p$ would have to wait to receive $y$ before it includes the round-$r$ famous witnesses in its hashgraph. Therefore, it would be impossible for $p$ to commit $x$ before it receives $y$. 

Now let us assume that $p$'s hashgraph contains both $x$ and $y$. If $x$'s roundReceived is determined when $p$ calls \emph{findOrder}, $y$'s will be as well. If $p$ assigns $r$ to $x.roundReceived$, all the round-$r$ famous witnesses that $p$ holds see $x$. By Lemma~\ref{lemma:late-witness}, any round-$r$ witness $p$ receives in the future will be declared not famous because it has already decided the famousness of at least one round-$r$ witness, so $p$ holds all round-$r$ famous witnesses. We already know that $y.roundReceived$ equals $r$ as well, so $p$ must assign $r$ to $y.roundReceived$ during the same call to \emph{findOrder}. This is because, by definition, every round-$r$ witness sees $y$ and any event some round-$r$ witness $w$ sees can be immediately found upon adding $w$ to $p$'s hashgraph.

We saw in Section~\ref{find-order} that once an event is assigned a roundReceived, we can immediately determine its location in the log by computing its logical timestamp and whitened signature. Thus, all events with the same roundReceived will be committed at the same time. An honest party can simply sort them by their logical timestamp and whitened signature before committing them to ensure that events are committed to the log in order.  

Now consider two events $x$ and $y$ such that $y.roundReceived < x.roundReceived$. If $p$ has committed $x$, it must have committed $y$ beforehand. Let $r' = y.roundReceived$ and $r = x.roundReceived$. 

By the definition of \emph{findOrder}, if $x.roundReceived$ is assigned $r$, then the famousness of all witnesses up to round $r$ must have been decided. Therefore, the famousness of all round-$r'$ witnesses that $p$'s hashgraph contains must have been decided. By Lemma~\ref{lemma:late-witness}, any other round-$r'$ witness $p$ receives in the future will be declared not famous. If $y$ has a roundReceived of $r'$, then all of the round-$r'$ famous witnesses that $p$ has will see $y$, which is something that can immediately be determined once $p$'s hashgraph contains all round-$r'$ famous witnesses. Therefore, it is impossible to assign $x$ a roundReceived before assigning $y$ a roundReceived. 

As we saw above, once an event is assigned a roundReceived, we can immediately compute its logical timestamp and whitened signature and find its location in the final log. Therefore, $y$ must be committed during the same call to \emph{findOrder} as when $x$ was committed or an earlier one. Once again, a party can simply sort the batch of events it is ready to commit after a call to \emph{findOrder} so that it can commit them in order. 

We have shown that all honest parties output events in sequence and in the same order as all other honest parties.
\end{proof}

\begin{theorem}
\label{liveness}
The HashGraph Protocol is live.
\end{theorem}

\begin{proof}
We will first show that every witness will eventually have its famousness decided and then show that any event $x$ created by an honest party will eventually be output to all honest parties' logs.

Consider some round-$r$ witness $y$ in some honest party $p$'s hashgraph, and let round $r_c > r$ be a coin round in $y$'s famousness decision. Suppose some round-$r_c$ witness $w$ does not determine its vote via coin flip. It must strongly see more than $2n/3$ round-$(r_c-1)$ witnesses that cast some vote $v$, which $w$ takes as its vote. Furthermore, if some other round-$r_c$ witness $w'$ does not determine its vote via coin flip, then $w'$ likewise strongly sees more than $2n/3$ round-$(r_c-1)$ witnesses that vote $v'$. Because more than $2n/3$ round-$(r_c-1)$ witnesses voted for $v$, fewer than $n/3$ must have voted for $\neg v$, and so $v = v'$. As a result, all witnesses that do not flip a coin during $r_c$ cast the same vote.

The remaining round-$r_c$ witnesses determine their vote via coin flip. Those coin flips may cause more than $2n/3$ round-$r_c$ witnesses to vote for $v$, and if this occurs, every round-$(r_c+1)$ witness will vote for $v$. Over many coin rounds, the probability that this occurs approaches 1. By Lemma~\ref{consistency:consistent-fame}, once there is a coin round $r'$ in which this happens, a decision is made within a constant number of rounds. Therefore, a witness' famousness must eventually be decided.

We now show that every event $x$ created by an honest party will eventually be an ancestor of all unique famous witnesses in some round. Every honest party eventually receives $x$ in a sync and then creates an event $y$ that is a descendant of $x$. $x$'s creator is honest and will never fork, so $y$ sees $x$ and any descendant of $y$ will see $x$. Thus, there is some round where all witnesses created by honest parties see $x$, so any honest witness declared famous in that round will see $x$.

It remains to show that any unique famous witnesses from a corrupted party will also eventually see $x$. Let round $r$ be the first round in which all famous witnesses created by honest parties see $x$. Any witness in round $(r+1)$ must strongly see---and so must see---more than $2n/3$ round-$r$ witnesses. There are fewer than $n/3$ corrupted parties, so every round-$(r+1)$ witness must strongly see at least one round-$r$ witness from an honest party. Therefore, every round-$(r+1)$ witness sees a witness that sees $x$. Because $x$'s creator never forks, all round-$(r+1)$ famous witnesses see $x$.

When $x$ is seen by all famous witnesses in some round of some party's hashgraph, it is output to the log that party maintains. Since we have just shown that this eventually occurs for any honest event, and because the adversary must eventually deliver all events, all honest parties will eventually output all honest events.
\end{proof}

\section{Delaying a Famousness Decision}\label{sec:delay}

By the FLP result \cite{flp}, it is impossible for an asynchronous BFT protocol to be deterministic. The HashGraph Protocol implements randomness in the procedure \emph{decideFame} through periodic coin rounds. In a coin round, witnesses may use the middle bit of their signature to cast their vote. This makes it impossible for a famousness decision to be delayed indefinitely as there is always a small chance that more than $2n/3$ witnesses vote the same way during a coin round, and once this occurs a decision is guaranteed within a constant number of rounds (shown in Lemma~\ref{consistency:consistent-fame}). Below we show a strategy by which an adversary in control of the network can delay a famousness decision until such a coin round occurs, forcing the decision to take an expected exponential amount of rounds. If we delay the famousness decision of some round-$r$ witness $x$ and some event $y$'s roundReceived is $r$ or later, then we delay committing $y$ to the log for at least as long as $x$'s fame is undecided. Therefore, this attack implies an expected exponential delay to liveness. We supply pseudocode for the attack in Appendix~\ref{sec:appendix}.

First, let us clarify what properties the attack must have. In every normal round, the attack should ensure that fewer than $2n/3$ witnesses vote ``yes" and fewer than $2n/3$ witnesses vote ``no". Furthermore, if a coin round occurs and it is still the case that fewer than $2n/3$ witnesses vote the same way, the subsequent normal rounds should not lead to a famousness decision. As stated above, the only way for a decision to occur is if a coin round causes more than $2n/3$ witnesses to vote the same way. We show in Appendix~\ref{sec:appendix-b} that if our attack has these parameters, then the adversary can sustain the attack for an exponential number of rounds in expectation.

Before we begin the proof, we will provide a brief motivation and explanation of our methods. Consider some round-$r$ witness $x$ whose famousness decision we wish to delay. As we construct some normal round $\rho > r+1$, our attack assumes that the round-$(\rho-1)$ witnesses have not produced a supermajority vote on $x$'s famousness (our strategy applies even if round $(\rho - 1)$ is a coin round). Thus, we can use induction to show that our attack is valid for any $\rho$.

We also assume that, when constructing round-$\rho$ witnesses, no round-$(\rho - 1)$ witness sees another round-$(\rho-1)$ witness. To illustrate why this is a concern, consider a round-$(\rho-1)$ witnesses $z_1$ that  sees a round-$(\rho-1)$ witnesses $z_2$. If a round-$\rho$ witness strongly sees $z_1$, it strongly sees $z_2$ as well, possibly disrupting the strategy. Our construction automatically prevents this from occuring when we manipulate the famousness votes in round $(\rho - 1)$.

If round $(\rho - 1)$ is a coin round, the adversary can use the method shown in the inductive step to ensure that no two round-$(\rho - 1)$ witnesses see each other despite the fact that the construction will not manipulate the witnesses' votes.

In addition to the assumption that round-$(\rho - 1)$ witnesses do not see each other, we rely on the fact that the adversary can influence the construction of honest hashgraphs. We now describe how this is done. The adversary takes advantage of two properties of the HashGraph Protocol to manipulate honest hashgraphs. First, it knows that honest parties create an event upon receiving a sync. Second, honest parties create a single event upon receiving a sync, even if that sync updates them on multiple events. 

For example, in Figure~\ref{fig:adversary-control}, the adversary has first delivered $p_4$'s genesis event sync to $p_3$, inducing $p_3$ to create an event. Then the adversary delivers $p_3$'s sync to $p_2$, causing $p_2$ to create an event whose self-parent is $p_3$'s event, the newest event in the sync. Note that eventually $p_2$ will receive $p_4$'s initial sync of its first event, but because $p_2$ already received that event from $p_3$, $p_2$ will discard it. The final state represents $p_2$'s hashgraph (except for $p_1$'s initial event, shown for clarity, which $p_2$ does not have). 

\begin{figure}[h]
\begin{center}
\begin{minipage}[t]{0.33\textwidth}
\centering
  %\includesvg[scale=0.5]{images/adversary control 3.svg}
  \includegraphics[scale=0.5]{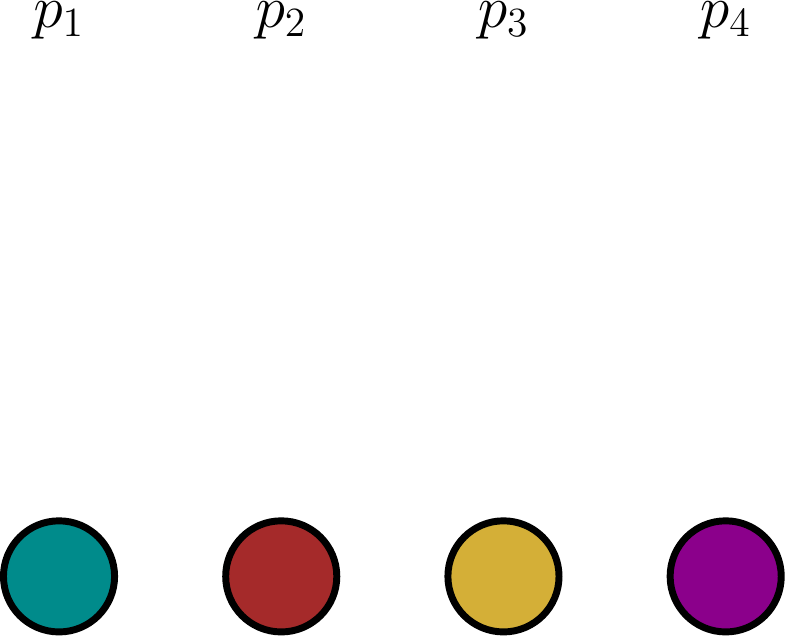}
\end{minipage}\hfill
\begin{minipage}[t]{0.33\textwidth}
\centering
  %\includesvg[scale=0.5]{images/adversary control 2.svg}
  \includegraphics[scale=0.5]{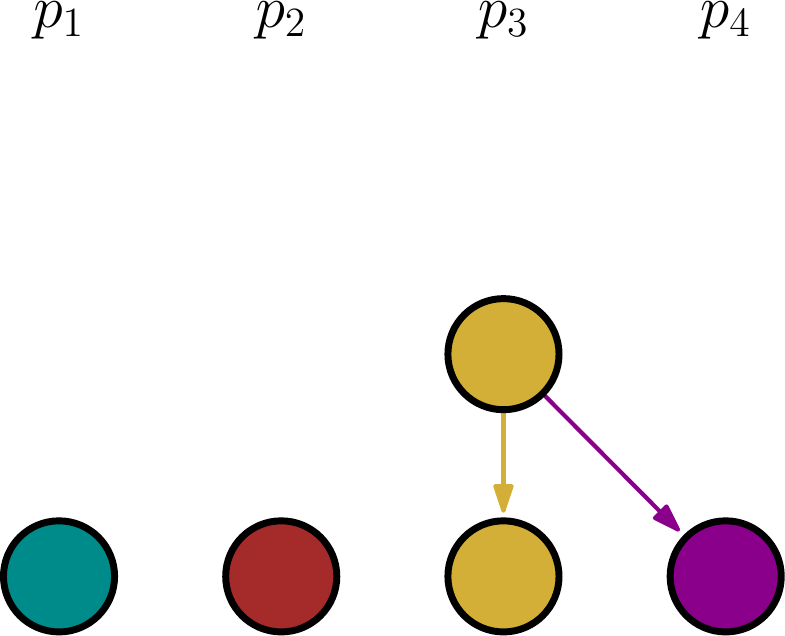}
\end{minipage}\hfill
\begin{minipage}[t]{0.33\textwidth}
\centering
      %\includesvg[scale=0.5]{images/adversary control 1.svg}
      \includegraphics[scale=0.5]{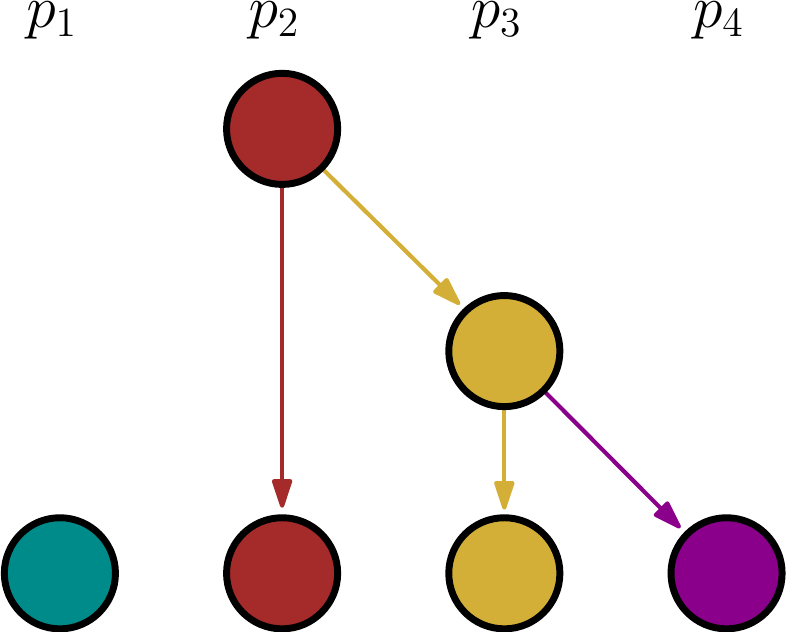}
\end{minipage}\hfill
\end{center}

\caption{The adversary influences a hashgraph's structure by selectively choosing which events to deliver.}\label{fig:adversary-control}

\end{figure}

Note that the strategy works even if all parties act honestly---all that is required of the adversary is its ability to choose the order in which to deliver syncs. In fact, forking would only hinder the strategy by making it more difficult to manipulate what events each round-$\rho$ witness sees.

By Lemma~\ref{lemma:graph-consistent} and the fact that the adversary eventually delivers all events, we know that honest parties will eventually have graph-consistent hashgraphs. Thus, if we induce a delay on $x$'s famousness decision for one honest hashgraph, it will carry through to all honest hashgraphs.

\begin{theorem}
Fix some round-$r$ witness $x$. There exists an adversarial strategy delaying a decision on x's famousness for exponentially many rounds in expectation. 
\label{thm:delay}
\end{theorem}

\begin{proof}

We will show that the adversary can manipulate the order in which parties receive events to ensure that fewer than $2n/3$ round-$(r+1)$ witnesses vote ``yes" and fewer than $2n/3$ round-$(r+1)$ witnesses vote ``no". In addition, we will show that in any round $r' > r+1$ that is not a coin round, the adversary can force a round-$r'$ witness from party $p$ to vote the same way as $p$'s round-$(r'-1)$ witness. It follows that a supermajority vote can only form in a coin round. As shown in Appendix~\ref{sec:appendix-b}, it takes an expected exponential number of coin rounds for this to occur. 

We assume for simplicity that $n \geq 7$. Any $n$ less than 7 has the property that the smallest integer greater than $2n/3$ ($\lfloor 2n/3 \rfloor + 1$) equals $n - 1$, requiring a slightly different strategy.

We first prove the base case, which explains our strategy for round-$(r+1)$ witnesses. Then we prove the inductive step, which handles round-$r'$ witnesses for $r' > r+1$. We describe the attack thoroughly for $n = 7$ and then show how it generalizes to larger $n$. When $n$ is 7, the quorum size is $\lfloor 2 * 7 / 3 \rfloor + 1 = 5$.

\begin{figure}[h]
\begin{center}
%\includesvg[scale=0.5]{images/Base Case 1.svg}
\includegraphics[scale=0.5]{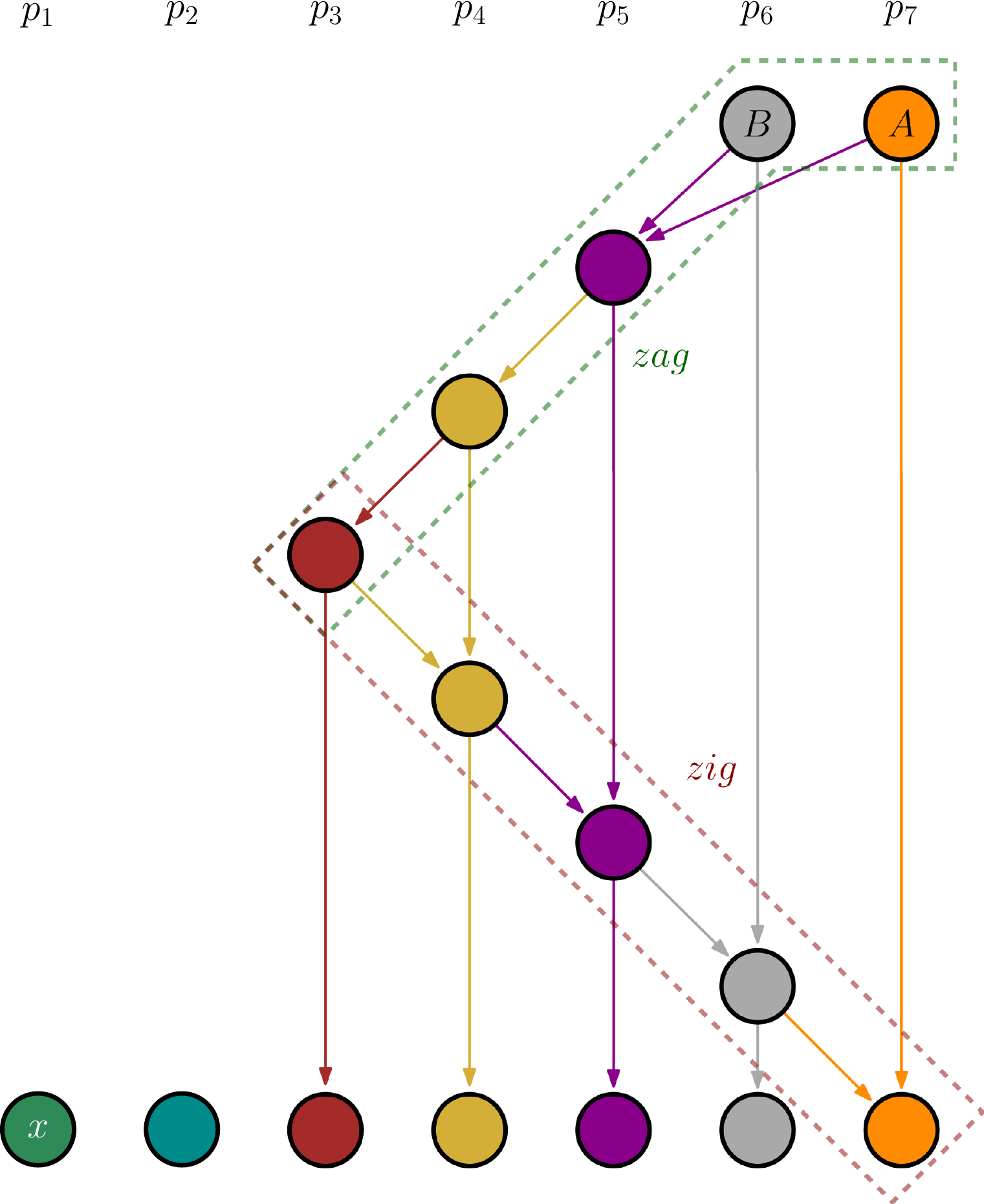}
\end{center}

\caption{Base Case: the \emph{zig} and \emph{zag} give events $A$ and $B$ the desired properties.}\label{fig:base-case-1}

\end{figure}

Now we will describe the base case. Let $x$ be a round-$r$ witness whose famousness the adversary wishes to delay. We assume that no round-$r$ witnesses see each other. By definition, a round-$(r+1)$ witness votes on $x$'s famousness based on whether it sees $x$. Thus, the adversary's goal is to design a hashgraph such that fewer than $\lfloor 2*7/3 \rfloor + 1 = 5$ round-$(r+1)$ witnesses see $x$ and fewer than 5 round-$(r+1)$ witnesses do not see $x$. 

We begin by assigning each party a number from 1 to 7, where $p_1$ is $x.creator$ and the rest are arbitrarily assigned. The strategy has three phases. In the first phase, we construct a hashgraph that culminates in 2 events, $A$ and $B$. (Figure~\ref{fig:base-case-1} depicts an example for $n=7$.) The up-left diagonal of 5 events in the construction is called the \emph{zig} and the subsequent up-right diagonal is called the \emph{zag}. 

Importantly, $A$ and $B$ are not round-$(r+1)$ witnesses: $A$ and $B$ only see events from 4 distinct parties in the zag, which is not enough to strongly see every event in the zig ($A$ only strongly sees $p_6$ and $p_7$'s round-$r$ witnesses and $B$ only strongly sees $p_7$'s). However, any event $y$ that sees them both strongly sees every event in the zig. 

Seeing $A$ and $B$ allows $y$ to see all 5 events from distinct parties in the zag, letting $y$ strongly see the 5 round-$r$ witnesses created by parties $p_3$ through $p_7$. This promotes $y$ to a higher round. 

Note that the figure does not depict a single hashgraph owned by a single party (if it were owned by a single party, then that party would have created an event upon receiving either $B$ or $A$ or both). Instead, it depicts the events and syncs that the adversary has delivered so far.

In Algorithm~\ref{algo:base-case}, lines 4 to 9 describe how we create the \emph{zig}, lines 11 to 22 describe how we create the \emph{zag}, and lines 17 to 22 describe how we create $A$ and $B$.

\begin{figure}[h]
\begin{center}
%\includesvg[scale=0.5]{images/Base Case 2.svg}
\includegraphics[scale=0.5]{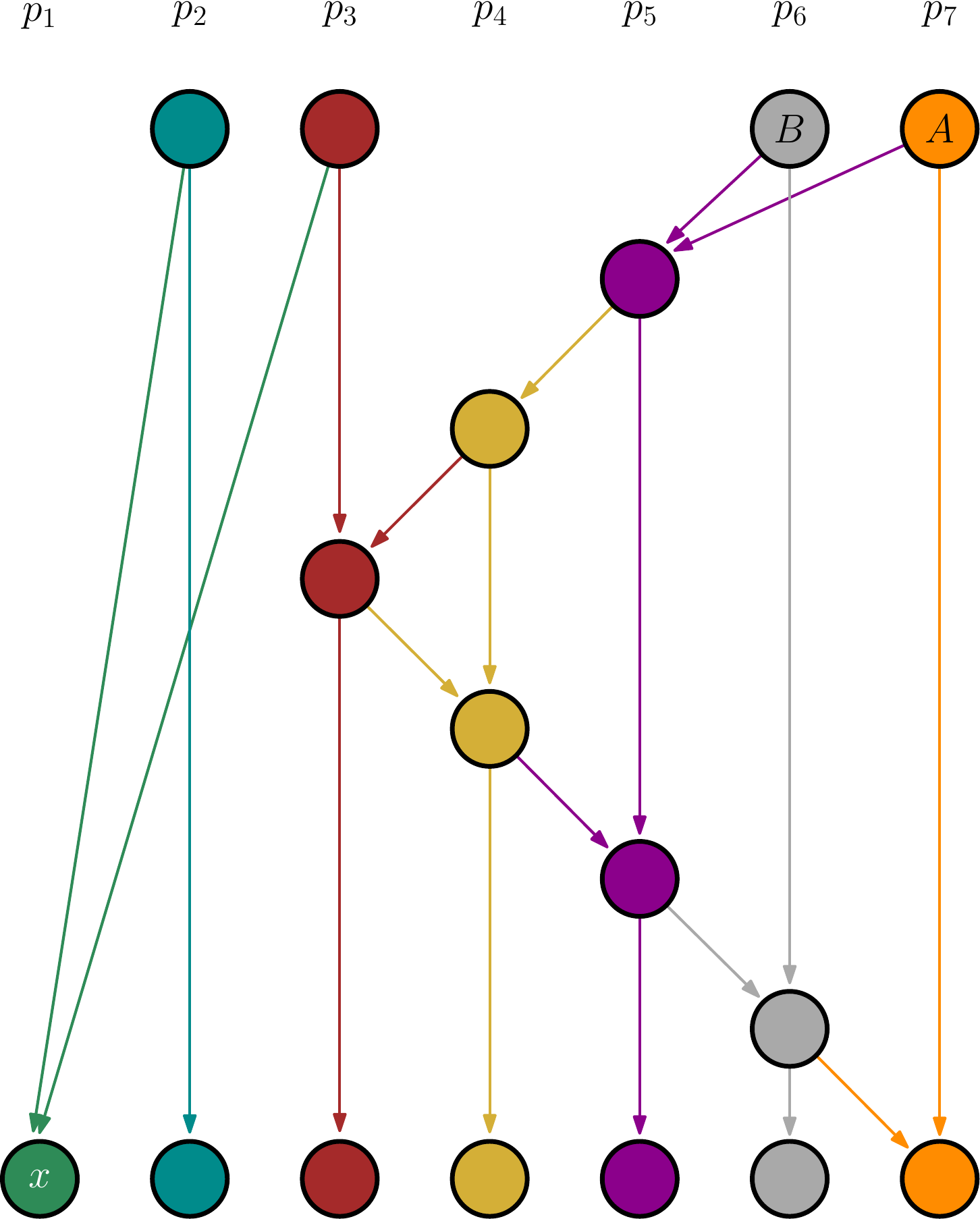}
\end{center}

\caption{Base Case: fixing the ``yes" votes.}\label{fig:base-case-2}

\end{figure}

Before we use $A$ and $B$ to create the round-$(r+1)$ witnesses, we must fix the witnesses' votes. This is phase two of the strategy. We would like 3 parties to vote ``yes" and 4 to vote ``no", so we deliver $x$ to parties $p_2$ and $p_3$, ensuring the leftmost 3 parties create round-$(r+1)$ witnesses that see $x$ and the rest do not. This is shown in Figure~\ref{fig:base-case-2} and detailed in lines 37 and 43 of Algorithm~\ref{algo:base-case}.

\begin{figure}[h]
\begin{center}
%\includesvg[scale=0.5]{images/Base Case 3.svg}
\includegraphics[scale=0.5]{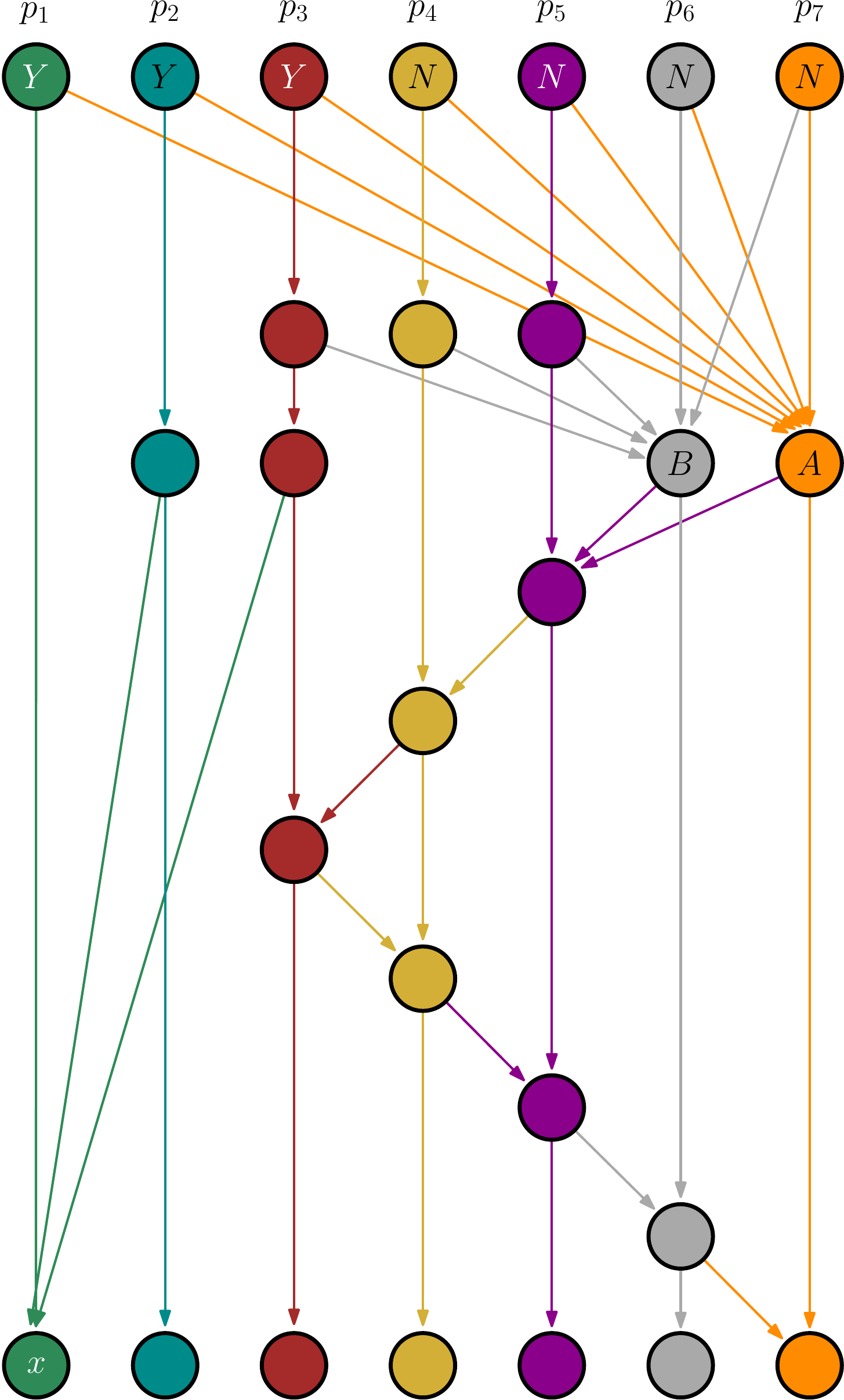}
\end{center}

\caption{Base Case: producing the round-$(r+1)$ witnesses.}\label{fig:base-case-3}

\end{figure}

Lastly, in phase three, we deliver $A$ and $B$ such that every party's round-$(r+1)$ witness strongly sees the rightmost five round-$r$ witnesses. For parties $p_1$ and $p_2$, who did not participate in creating the zig and zag, it is sufficient to receive either $A$ or $B$ (we choose to deliver $A$ arbitrarily). This is because an event from $p_1$ or $p_2$ sees itself, adding an extra count to the number of events from distinct parties that it sees. Also, the parties that created $A$ and $B$ of course only need to receive $B$ and $A$ respectively. 

We now have a set of $7$ round-$(r+1)$ witnesses such that 3 vote ``yes", 4 vote ``no", and none see each other, completing our description of the base case. This is visualized in Figure~\ref{fig:base-case-3} and lines 27 to 47 of Algorithm~\ref{algo:base-case} describe this phase. 

Generalizing this strategy to $n > 7$ is simple. In phase one, rather than being composed of events from parties $p_3$ through $p_7$, the zig and zag are composed of events from the rightmost $\lfloor 2n/3 \rfloor + 1$ parties. Note that, for any $n$, because the zig and zag are only $\lfloor 2n/3 \rfloor + 1$ parties long, every event in the zig will strongly see either nothing, or $p_n$'s witness. Every event in the zag will strongly see $p_n$'s witness except for $A$, which also strongly sees $p_{n-1}$'s witness. Thus, none of the events in the zig and zag are round-$(r+1)$ witnesses. 

In phase two, instead of delivering $x$ to parties $p_2$ and $p_3$, we deliver $x$ to parties $p_2$ through $p_{\lfloor n/2 \rfloor}$, ensuring that $\lfloor n/2 \rfloor$ parties vote ``yes" and the remaining parties vote ``no". 

In phase three, we again distinguish between the parties that contributed to the zig and zag and those that did not. We deliver $A$ and $B$ to the parties that contributed to the zig and zag, the rightmost $\lfloor 2n/3 \rfloor + 1$ parties. An event $y$ created by one of those parties that sees $A$ and $B$ sees $\lfloor 2n/3 \rfloor + 1$ events from distinct parties that each see all the events in the zig. Thus, $y$ strongly sees the rightmost $\lfloor 2n/3 \rfloor + 1$ round-$r$ witnesses, promoting it to a round-$(r+1)$ witness. The parties that did not contribute to the zig and zag only receive $A$, as the event they produce upon receiving $A$ adds one to the count of events from distinct parties. Combined with the $\lfloor 2n/3 \rfloor$ events from distinct parties that $A$ alone sees in the zag, this is sufficient to strongly see the rightmost $\lfloor 2n/3 \rfloor + 1$ round-$r$ witnesses.

Now we describe the inductive step for $n=7$. Consider some round $r' > r+1$ that is not a coin round. Recall that  round-$r'$ witnesses vote by taking the majority vote of the quorum of round-$(r'-1)$ witnesses they strongly see. This requires a slightly different adversarial strategy from the base case. We once again assume that no round-$(r'-1)$ witnesses see each other, and we assume our inductive hypothesis: fewer than 5 round-$(r'-1)$ witnesses voted ``yes" and fewer than 5 round-$(r'-1)$ witnesses voted ``no".

Let $v$ be the majority vote of all the round-$(r'-1)$ witnesses and $\neg v$ be the minority vote (in the case of a tie, $v$ is yes). Let $c_v$ be the number of parties whose round-$(r'-1)$ witnesses voted $v$ and $c_{\neg v}$ be the number of parties whose round-$(r'-1)$ witnesses voted $\neg v$. Once again we number each party from 1 to 7. This time, $p_1$ to $p_{c_v}$ are parties whose round-$(r'-1)$ witnesses voted for $v$ and $p_{c_v+1}$ to $p_7$ are parties whose round-$(r'-1)$ witnesses voted $\neg v$. 

Our inductive hypothesis states that $c_v$ is strictly less than 5, so $c_{\neg v}$ is strictly more than 2. This allows us to show that there always exists a group of 5 round-$(r' - 1)$ witnesses with a majority vote for $\neg v$. If $c_{\neg v}$ were 3, for example, a quorum with a majority vote of $\neg v$ exists by grouping the at least 3 $\neg v$-voters with 2 $v$-voters. In general, we create the quorum of events with majority vote $\neg v$ by packing it with all of the $\neg v$ voters and $(5 - c_{\neg v})$ $v$-voters. Because $c_{\neg v}$ is at least 3, this always leads to a majority vote for $\neg v$. Consequently, if we force parties $p_{c_v + 1}$ to $p_7$'s round-$r'$ witnesses to strongly see such a quorum (and no other round-$(r'-1)$ witnesses), then their round-$r'$ witnesses will vote $\neg v$, just like their round-$(r'-1)$ witnesses did. If we force the remaining round-$r'$ witnesses from parties $p_1$ to $p_{c_v}$ to strongly see a quorum of all round-$(r'-1)$ witnesses, then their round-$r'$ witnesses will vote $v$, just like their round-$(r'-1)$ witnesses did.

\begin{figure}[h]
\begin{center}
    %\includesvg[scale=0.5]{images/Inductive Step 1.svg}
    \includegraphics[scale=0.5]{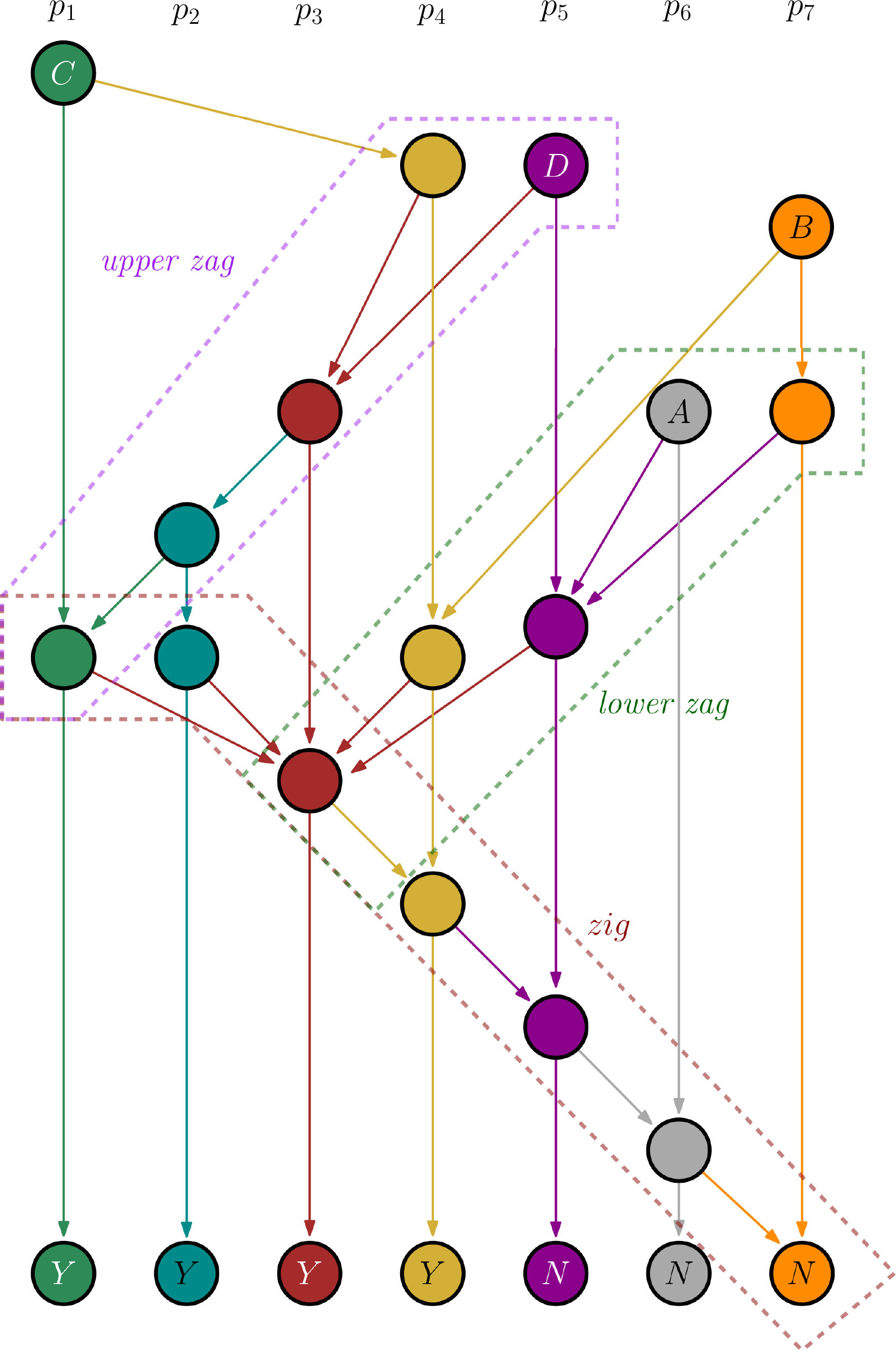}
\end{center}

\caption{Inductive Step: The zig and the lower zag give $A$ and $B$ their desired properties. The zig and upper zag do the same for $C$ and $D$.}\label{fig:inductive-step-1}

\end{figure}

Our approach is similar to the base case, but here we have two phases rather than three; the quorum a round-$r'$ witness strongly sees decides its vote, so we do not need a separate phase to assign votes like we did in the base case. We have visualized the first phase of the inductive step in Figure~\ref{fig:inductive-step-1}. There are two main differences: now the zig spans all 7 parties. It also is split at the end: $p_1$'s event in the zig points to $p_3$'s event rather than $p_2$'s event. If the zig were continuous, $D$ would strongly see the rightmost 5 round-$(r'-1)$ witnesses and would be a round-$r'$ witness. This would be a problem because our strategy promotes events to round $r'$ in part by having them see $D$ and, if $D$ were a round-$r'$ witness, that would violate the assumption that no witnesses in the same round see each other. As for the zags, there are two now instead of one: the \emph{lower zag}, which is similar to the base case's zag, as well as the \emph{upper zag}. 5 parties contribute to each zag like in the base case. 

Events $A$ and $B$ function like they do in the base case: both are in round $(r'-1)$ as neither strongly see more than 2 round-$(r'-1)$ witnesses, but an event that sees both $A$ and $B$ strongly sees the rightmost 5 round-$(r'-1)$ witnesses. Those 5 round-$(r'-1)$ witnesses contain all $c_{\neg v}$ $\neg v$-voters, so any round-$r'$ witness created upon its party receiving $A$ and $B$ will vote for $\neg v$.

Events $C$ and $D$ are also both in round $(r'-1)$ for similar reasons. An event that sees them both strongly sees all round-$(r'-1)$ witnesses, so a round-$r'$ witness that sees $C$ and $D$ votes for $v$. The reason that $C$ was created by $p_1$ rather than $p_4$ is to accommodate the split end of the zig described above. Consider if $C$ were instead $p_4$'s event in the upper zag. Then an event created by either $p_3$, $p_4$, or $p_5$ that sees $C$ and $D$ would not strongly see $p_2$'s round-$(r'-1)$ witness as it would not see any event from $p_1$ before it sees $p_2$'s round-$(r'-1)$ witness. To fix this we assign $C$ to be an event created by $p_1$ that sees $p_4$'s upper zag event.

Also note that $p_5$'s event in the lower zag skips over $p_4$'s event. This occurs because the lower zag overlaps with the upper zag and is once again a problem concerning the events $D$ sees. If $p_5$'s event did not skip over $p_4$'s event, then $D$ would end up seeing an event from party $p_4$ and would be a round-$r'$ witness. This means that we need either $A$ or $B$ to see $p_4$'s event in the lower zag. We chose $B$ for this arbitrarily, requiring us to deliver $p_4$'s event to $p_7$ to induce it to create $B$.

Lines 8 to 16 in Algorithm~\ref{algo:inductive-step} describe how to create the zig. Lines 18 to 32 of the algorithm describe the lower zag and create $A$ and $B$. Lines 34 to 47 describe the upper zag and create $C$ and $D$.

In the second phase, all that remains is to ensure that parties $p_1$ to $p_{c_v}$ create events that see $C$ and $D$ and parties $p_{c_v+1}$ to $p_7$ create events that see $A$ and $B$. Figure~\ref{fig:inductive-step-2} depicts the completed construction for $n=7$. Lines 49 to 57 of Algorithm~\ref{algo:inductive-step} show how we create the majority- and minority-vote round-$r'$ witnesses. 

\begin{figure}[h!]
\begin{center}
%\includesvg[scale=0.5]{images/Inductive Step 2.svg}
\includegraphics[scale=0.5]{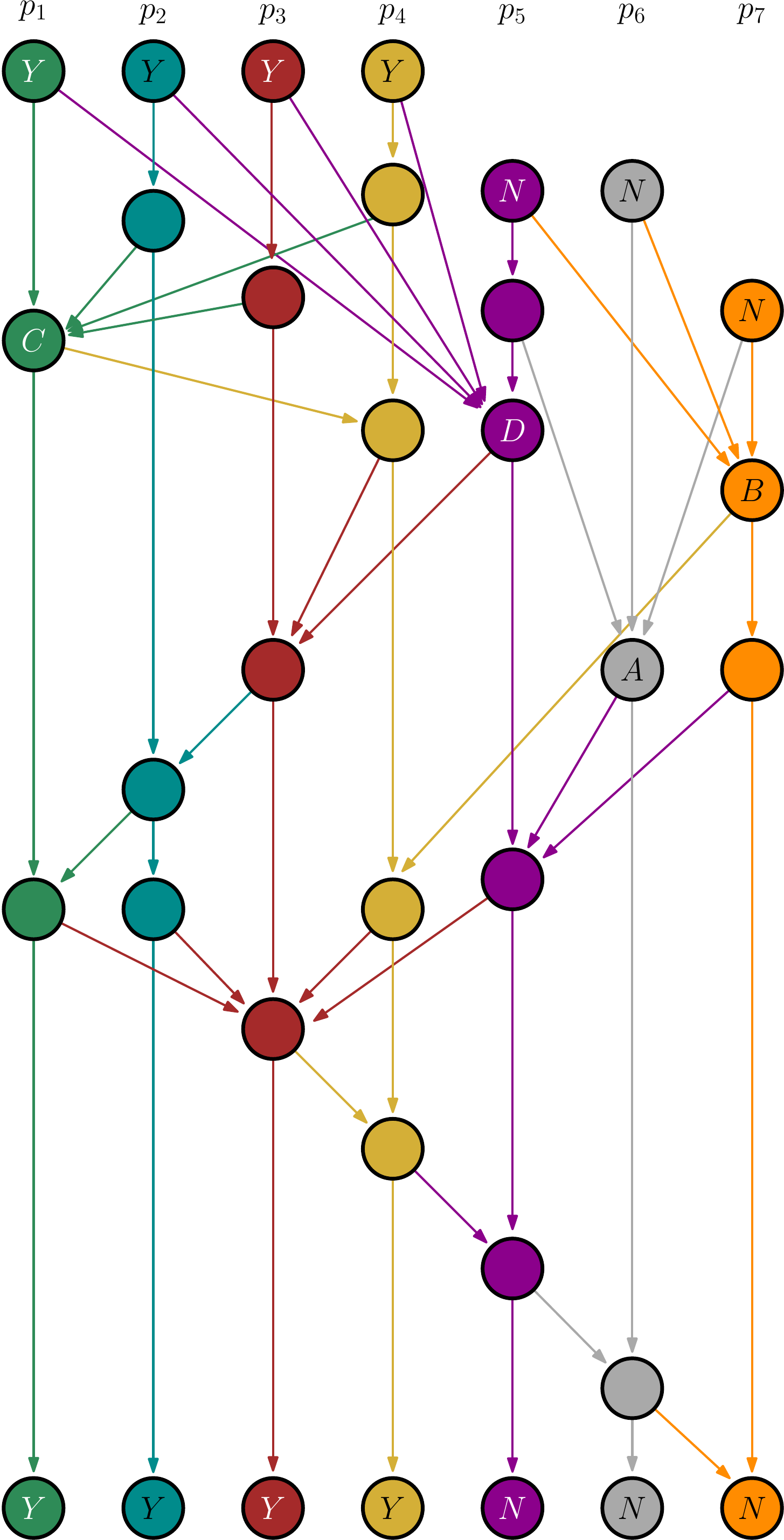}
\end{center}

\caption{Inductive Step: producing the round-$r'$ witnesses.}\label{fig:inductive-step-2}

\end{figure}

Once again, generalizing to larger $n$ is straightforward. The zig will always include all $n$ parties. The upper zag starts at $p_1$ and ends at $p_{\lfloor 2n/3 \rfloor + 1}$ while the lower zag starts at $p_{n - \lfloor 2n/3 \rfloor}$ and ends at $p_n$. In other words, the upper zag involves the first $\lfloor 2n/3 \rfloor + 1$ parties whereas the lower zag involves the last $\lfloor 2n/3 \rfloor + 1$ parties. 

In the general case, the zig is still split at the end to prevent $D$ from being a witness. Thus, $C$ is an event by $p_1$ that sees $p_{\lfloor 2n/3 \rfloor}$'s upper zag event. We also must accommodate the upper and lower zags' overlap to prevent $D$ from being a witness: $p_{\lfloor 2n/3 \rfloor +1}$'s event points to $p_{\lfloor 2n/3 \rfloor - 1}$'s event in the lower zag rather than $p_{\lfloor 2n/3 \rfloor}$'s event. As a result, $B$ is created after the lower zag has formed by delivering $p_{\lfloor 2n/3 \rfloor}$'s event to $p_n$.

Creating the round-$r'$ witnesses works the same as in the $n=7$ case.

We have shown that it is possible to force round-$r'$ witnesses into a deadlock by manipulating the order in which parties receive events. Because no supermajority vote exists upon entering the coin round, all parties flip coins. As long as the coin flips do not cause more than $2n/3$ parties vote the same way, our inductive hypothesis will still hold. We showed in Appendix~\ref{sec:appendix-b} that under these circumstances, the adversary can sustain the attack for an expected exponential number of rounds.

\end{proof}
\section{Discussion}\label{sec:discussion}
Above we have presented a description of the HashGraph Protocol and a formal proof of its security. Hopefully this will make it easier to compare this protocol with other atomic broadcast constructions and extend its ideas in future work. We also showed a strategy by which an adversary can stall a famousness decision for exponentially many rounds in expectation, proving that the upper bound on the round-complexity of the \emph{decideFame} procedure presented in the original paper \cite{swirlds} is tight. 

The algorithm described in Theorem~\ref{thm:delay} currently only allows the adversary to delay a single famousness decision at a time. The base case algorithm forces all round-$(r+1)$ witnesses to see the rightmost $\lfloor 2n/3 \rfloor + 1$ witnesses, meaning that the round-$(r+1)$ witnesses would unanimously vote yes in a famousness decision for those witnesses. In the inductive step, every round-$r'$ witness is seen by at least $\lfloor 2n/3 \rfloor +1$ round-$(r'+1)$ witnesses. Therefore, every witness $w$ in round $r' \geq r+1$ will be declared famous for as long as the adversary can delay $x$'s famousness decision. It would be interesting to consider the implications of this and its usefulness to the adversary. Another interesting direction is whether other DAG-based constructions might be susceptible to attacks like these, and what measures can be put in place to make this line of attack infeasible. 

\bibliographystyle{plain}
\bibliography{abbrev1,refs}

\pagebreak
% Ensure that single-page floats stay at top.
\makeatletter
\setlength{\@fptop}{0pt}
\setlength{\@fpbot}{0pt plus 1fil}
\makeatother

\appendix
\section{Supplementary Algorithms}\label{sec:appendix}

Below we reproduce the HashGraph Protocol's three local functions: \emph{divideRounds}, \emph{decideFame}, and \emph{findOrder}. We also include the pseudocode for the base case and inductive step of the attack presented in Section~\ref{sec:delay}.

\begin{algorithm}
\begin{algorithmic}[1]
\caption{divideRounds \protect{\cite{swirlds}}}
\STATE \textbf{Input: }$c$: how frequently coin rounds occur.
\FOR{each event $x$}
    \STATE $r \leftarrow$ max round of parents of $x$ (or 1 if none exist)
    \IF{$x$ can strongly see more than $2n/3$ round-$r$ witnesses}
        \STATE $x.round \leftarrow r + 1$
    \ELSE {}
        \STATE $x.round \leftarrow r$
    \ENDIF
    \STATE $x.witness \leftarrow $ ($x$ has no self-parent) or ($x.round > x.selfParent.round$)
\ENDFOR
\end{algorithmic}
\end{algorithm}

\begin{algorithm}
\begin{algorithmic}[1]
\caption{decideFame \protect{\cite{swirlds}}}
\FOR{each event $x$ in order from earlier rounds to later}
    \STATE $x.famous \leftarrow $ UNDECIDED
    \FOR{each event $y$ in order from earlier rounds to later}
        \IF{$x.witness$ \textbf{and} $y.witness$ \textbf{and} $y.round > x.round$}
            \STATE $d \leftarrow y.round - x.round$
            \STATE $s \leftarrow $ the set of witness events in round $y.round-1$ that $y$ can strongly see
            \STATE $v \leftarrow $ majority vote in $s$ (is TRUE for a tie)
            \STATE $t \leftarrow$ the number of events in $s$ with a vote of $v$
            \IF{d = 1} // first round of the election 
                \STATE $y.vote \leftarrow$ can $y$ see $x$?
            \ELSE
                \IF{$d \mod c > 0$} // this is a normal round
                    \IF{$t > 2n/3$} // if supermajority, then decide
                        \STATE $x.famous \leftarrow v$
                        \STATE $y.vote \leftarrow v$
                        \STATE break out of the $y$ loop
                    \ELSE {} // else, just vote
                        \STATE $y.vote \leftarrow v$
                    \ENDIF
                
                \ELSE {} // this is a coin round
                    \IF {$t > 2n/3$} // if supermajority, then vote
                        \STATE $y.vote \leftarrow v$
                    \ELSE
                        \STATE $y.vote \leftarrow$ middle bit of $y.signature$
                    \ENDIF
                \ENDIF
            
            \ENDIF
        \ENDIF
    \ENDFOR
\ENDFOR
\end{algorithmic}
\end{algorithm}

\begin{algorithm}
\begin{algorithmic}[1]
\caption{findOrder \protect{\cite{swirlds}}}
\FOR{each event $x$}
    \IF{there is a round $r$ such that there is no event $y$ in or before round $r$ that has $y.witness =$ TRUE and $y.famous =$ UNDECIDED
    \\\hspace{.5cm}\textbf{and} $x$ is an ancestor of every round-$r$ unique famous witness 
    \\\hspace{.5cm}\textbf{and} this is not true of any round earlier than $r$}
        \STATE $x.roundReceived \leftarrow r$
        \STATE $s \leftarrow$ set of each event $z$ such that $z$ is a self-ancestor of a round-$r$ unique famous witness, \\\hspace{1.25cm}and $x$ is an ancestor of $z$ but not of the self-parent of $z$
        \STATE $x.consensusTimestamp \leftarrow$ median of the timestamps of all events in $s$
    \ENDIF    
\ENDFOR

\STATE return all events that have $roundReceived$ not UNDECIDED, sorted by $roundReceived$, then ties sorted by consensus timestamp, then by whitened signature.
\end{algorithmic}
\end{algorithm}

\clearpage

\begin{algorithm}[H]
\caption{Base Case}\label{algo:base-case}
\begin{algorithmic}[1]
\STATE \textbf{Input:} A round-$r$ witness $x$
\STATE Assign a unique label to each party such that $x$'s creator is labeled $p_1$ and the others are arbitrarily labeled $p_2, \dots, p_n$. 
\STATE Let $q \leftarrow \lfloor 2n/3 \rfloor + 1$.
\STATE // Make the zig.
\STATE Let $zigEnd \leftarrow n - q + 1$ // events from the rightmost $q$ parties form the zig.
\STATE Let $e$ be the round-$r$ witness created by $p_n$.
\FOR {$i = n-1$ down to $zigEnd$}
    \STATE Allow $p_i$ to receive $e$. 
    \STATE Let $e \leftarrow$ the event $p_i$ creates upon receiving $e$.
\ENDFOR

\STATE
\STATE // Make the zag.
\STATE Let $e$ be the event created by $p_{zigEnd}$ in the loop above.
\STATE Let $zagStart \leftarrow zigEnd$.
\FOR {$i = zagStart + 1$ to $n-2$}
    \STATE Allow $p_i$ to receive $e$.
    \STATE Let $e \leftarrow$ the event $p_i$ creates upon receiving $e$.
\ENDFOR
\STATE Let $C \leftarrow e$.

\STATE
\STATE Allow $p_{n-1}$ to receive $C$.
\STATE Let $B \leftarrow$ the event $p_{n-1}$ creates upon receiving $C$.
\STATE Allow $p_n$ to receive $C$.
\STATE Let $A \leftarrow$ the event $p_n$ creates upon receiving $C$.
\STATE
\STATE // Fix the votes and create round-$(r+1)$ witnesses.
\STATE Allow $p_{n-1}$ to receive $A$ and $p_n$ to receive $B$.
\STATE
\STATE // The rightmost $\lfloor n/2 \rfloor + 1$ round-$(r+1)$ witnesses should not see $x$.
\STATE // These parties must all be part of the zig and zag, so they must receive $A$ and $B$ to create round-$(r+1)$ witnesses.
\FOR{$i = \lfloor n/2 \rfloor + 1$ to $n-2$}
    \STATE // $p_i$'s witness will vote no.
    \STATE Allow $p_i$ to receive $A$.
    \STATE Allow $p_i$ to receive $B$.
\ENDFOR

\STATE
\STATE // Round-$(r+1)$ witnesses that see $x$ whose parties did not contribute to the zig and zag.
\STATE // Sufficient for them to see just $A$ or $B$.
\FOR {$i = 2$ to $zagStart - 1$}
    \STATE Allow $p_i$ to receive $x$. // $p_i$'s witness will vote yes.
    \STATE Allow $p_i$ to receive $A$.
\ENDFOR

\STATE
\STATE // Round-$(r+1)$ witnesses that see $x$ whose parties did contribute to the zig and zag.
\STATE // Must see $A$ and $B$.
\FOR {$i = zagStart$ to $\lfloor n/2 \rfloor$}
    \STATE Allow $p_i$ to receive $x$. // $p_i$'s witness will vote yes.
    \STATE Allow $p_i$ to receive $A$.
    \STATE Allow $p_i$ to receive $B$.
\ENDFOR

\STATE
\STATE Allow $p_1$ to receive $A$.

\end{algorithmic}
\end{algorithm}

\begin{figure}[h]

\begin{center}

%\includesvg[scale=.67]{images/Base Case.svg}
\includegraphics[scale=.67]{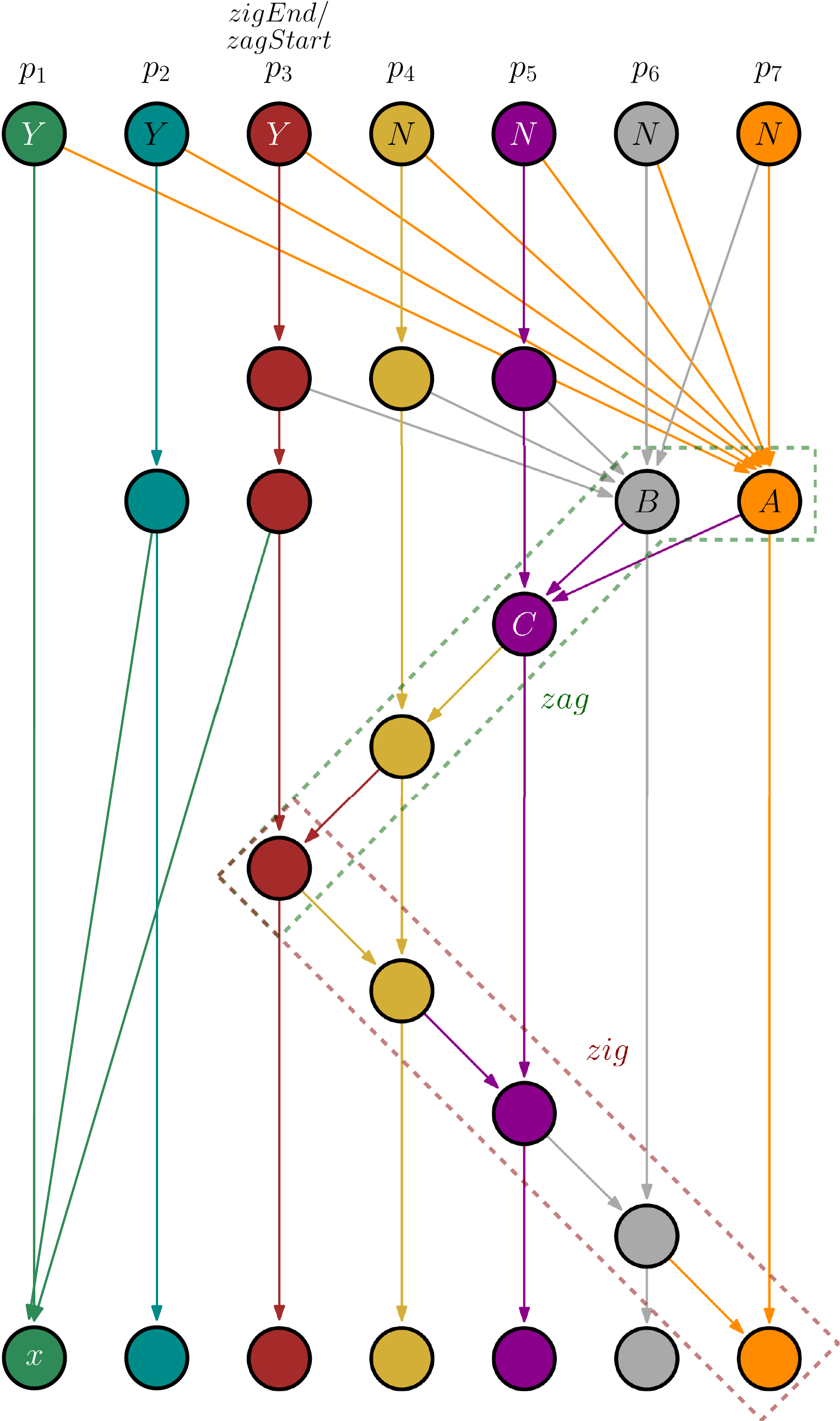}
\end{center}

\caption{Base Case Diagram for $n=7$.}\label{fig:base-case-appendix}

\end{figure}

\begin{algorithm}[H]
\caption{Inductive Step}\label{algo:inductive-step}
\begin{algorithmic}[1]
\STATE \textbf{Input:} The set of round-$(r'-1)$ witnesses and their votes on $x$'s famousness.
\STATE Let $v$ be the majority vote of the round-$(r'-1)$ witnesses ($v \leftarrow yes$ if there is a tie).
\STATE Let $\neg v$ be the minority vote.
\STATE Let $c_v$ be the number of majority votes.
\STATE Assign a unique label $p_1$ to $p_n$ to each party such that $p_1, \dots, p_{c_v}$ are assigned to parties whose round-$(r'-1)$ witnesses voted $v$ and $p_{c_{v}+1},\dots,p_n$ are assigned to parties whose round-$(r'-1)$ witnesses voted $\neg v$.
\STATE Let $q \leftarrow \lfloor2n/3\rfloor + 1$.
\STATE
\STATE // Create the zig.
\STATE Let $e \leftarrow p_n$'s round-$(r'-1)$ witness.
\FOR{$i = n-1$ down to $2$}
    \STATE Allow $p_i$ to receive $e$.
    \STATE Let $e \leftarrow$ the event $p_i$ creates upon receiving $e$.
\ENDFOR

\STATE
\STATE // Create the zig's forked end.
\STATE Allow $p_1$ to receive the event $p_3$ created in the loop above.
\STATE Let $e \leftarrow$ the event $p_1$ creates upon receiving $e$.

\STATE
\STATE // Create the lower zag.
\STATE Let $lowerStart \leftarrow n - q + 1$. // The lower zag spans the rightmost $q$ parties.
\STATE Let $e \leftarrow$ the event $p_{lowerStart}$ created in the zig.
\FOR{ $i=lowerStart + 1$ to $n-1$ }
    \STATE Allow $p_i$ to receive $e$.
    \IF { $i \neq q-1$ }
        \STATE Let $e \leftarrow$ the event $p_i$ creates upon receiving $e$.
    \ENDIF
\ENDFOR

\STATE
\STATE Let $e \leftarrow$ the event $p_{n-2}$ created in the loop above.

\STATE Allow $p_n$ to receive e.

\STATE
\STATE // Set up minority votes.
\STATE Let $A \leftarrow$ the event created by $p_{n-1}$ in the loop above.
\STATE Allow $p_n$ to receive the event created by $p_{q-1}$ in the loop above.
\STATE Let $B \leftarrow$ the event $p_n$ creates upon receiving that event.

\STATE
\STATE // Create the upper zag. It spans the leftmost $q$ parties.
\FOR {$i = 2$ to $q-1$}
    \STATE Allow party $p_i$ to receive $e$.
    \IF {$i \neq q-1$ } // Ensure that $p_q$'s upper zag event does not see $p_{q-1}$ so it is not a witness.
        \STATE Let $e \leftarrow$ the event $p_i$ creates upon receiving $e$.
    \ELSE
        \STATE Let $c \leftarrow$ the event $p_i$ creates upon receiving $e$.
    \ENDIF
\ENDFOR

\STATE
\STATE // Set up majority votes.
\STATE Allow $p_q$ to receive $e$.
\STATE // $D$ is not a witness because it does not see $p_{q-1}$'s upper zag event.
\STATE Let $D \leftarrow$ the event $p_q$ creates upon receiving $e$.
\STATE Allow $p_1$ to receive $c$.
\STATE Let $C \leftarrow$ the event $p_1$ creates upon receiving $c$.
\algstore{inductivestep}
\end{algorithmic}
\end{algorithm}

\begin{algorithm}[t]
\begin{algorithmic}
\algrestore{inductivestep}
\STATE // Fix the witnesses
\STATE // Create minority-vote witnesses.
\FOR { $i = c_v + 1$ to $n$ }
    \STATE if $p_i$ did not create $A$, then allow $p_i$ to receive $A$.
    \STATE if $p_i$ did not create $B$, then allow $p_i$ to receive $B$.
\ENDFOR
\STATE // Create majority-vote witnesses.
\FOR { $i = 1$ to $c_v$ }
    \STATE if $p_i$ did not create $C$, allow $p_i$ to receive $C$.
    \STATE if $p_i$ did not create $D$, allow $p_i$ to receive $D$.
\ENDFOR

\end{algorithmic}
\end{algorithm}

\begin{figure}[h]
\begin{center}
%\includesvg[scale=.67]{images/Inductive Step.svg}
\includegraphics[scale=.67]{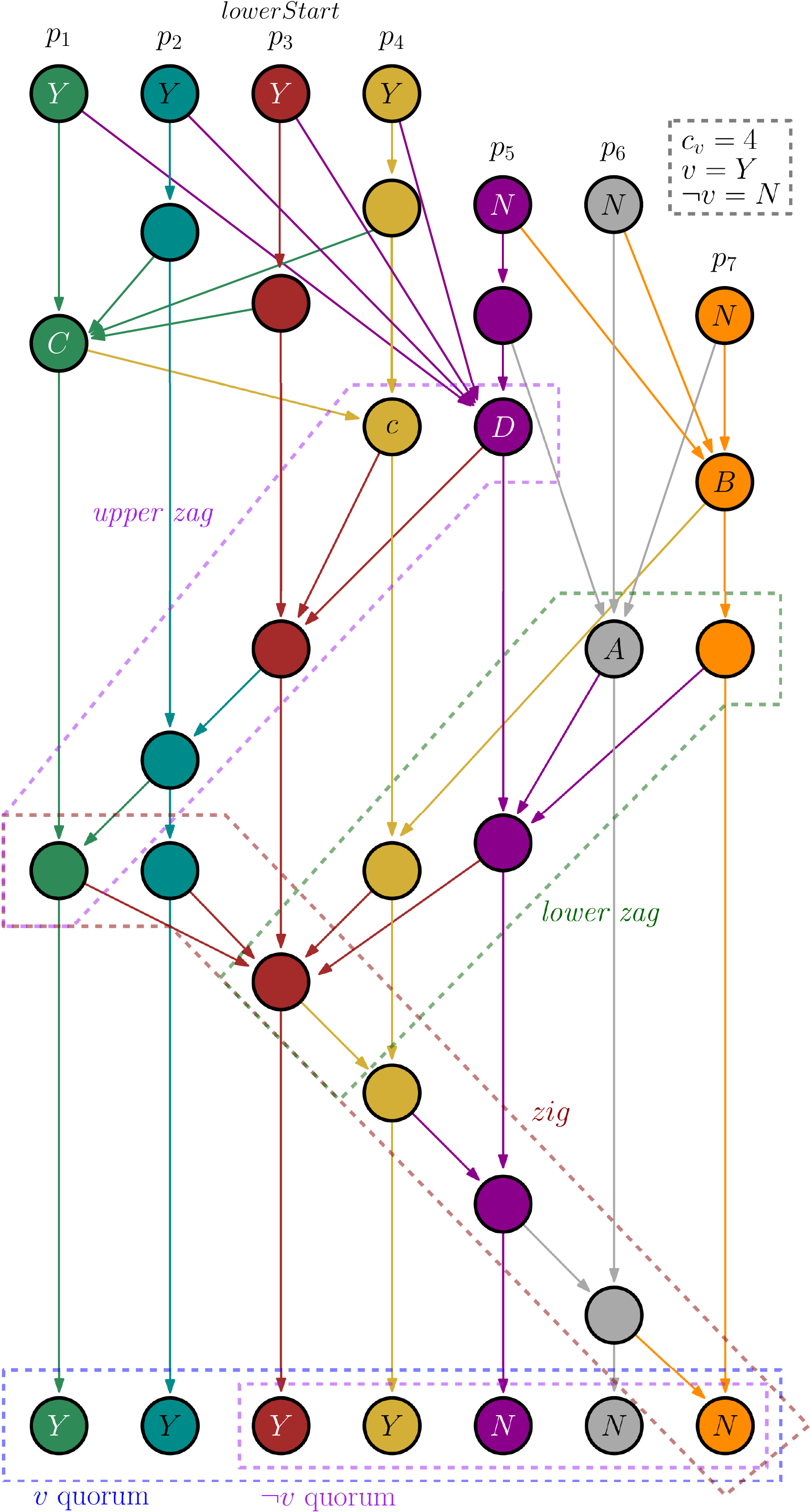}
\end{center}

\caption{Inductive Step Diagram for $n=7$.}\label{fig:inductive-step-appendix}

\end{figure}
%\pagebreak
\clearpage
\section{Proof of Delay Attack Round Complexity}\label{sec:appendix-b}

By Lemma~\ref{consistency:consistent-fame}, once more than $2n/3$ round-$r'$ witnesses vote for some $v$, it takes a constant number of rounds for a decision to occur. We showed in Section~\ref{sec:delay} that, if no supermajority vote for $v$ exists in round $(r'-1)$ and round $r'$ is a normal round, then we can guarantee that no supermajority vote for $v$ exists in round $r'$. As a result, the only way a supermajority could arise is during a coin round. 

Consider some coin round $r_c$ that occurs during the attack and assume that no supermajority exists among the round-$r_c - 1$ witnesses. That means no round-$r_c$ witness strongly sees more than $2n/3$ round-$(r_c - 1)$ witnesses that voted the same way. Therefore, all round-$r_c$ witnesses vote by coin flip. A supermajority arises in round $r_c$ if more than $2n/3$ coin flips agree on the same vote. We wish to show that, under these circumstances, it takes an expected exponential number of coin rounds for this to occur.

A witness flips a coin by examining the middle bit of its signature. We assume here that the result of that coin flip is uniform and independent and that the signature scheme produces a unique signature for each input. This means that the adversary cannot cheat by signing its events until it produces a signature with the middle bit it desires. 

Let us define some random variables that we will use to calculate the probability that a supermajority arises after the round-$r_c$ witnesses flip coins. For now, we ignore regular rounds. Let $X_{i,r}$ be a random variable that records the result of $p_i$'s round-$r$ witness' coin flip. Let $X_r$ be the sum of $X_{i,r}$ for $i = 1$ to $n$, the total number of ``yes" votes. Our strategy guarantees that witnesses from every party will vote in the election, so we can ignore the possibility of a witness being absent.

Let $Z_r$ be an indicator random variable that equals 1 if $X_r > 2n/3$ or $X_r < n/3$ and equals 0 otherwise. In other words $Z_r$ equals 1 if a supermajority of ``yes" votes or of ``no" votes arises in coin round $r$.

Let $R$ be a random variable that represents the first coin round in which $Z_r$ equals 1.

Therefore, proving that $E[R]$ is exponential in $n$ shows that it takes an expected exponential number of coin rounds for a supermajority vote to appear as long as the previous round did not have a supermajority vote.

\begin{claim}
$E[R]$ is exponential in $n$.
\end{claim}
\begin{proof}
First, we note that $Z_r$ is a Bernoulli random variable in which $p = Pr[X_r > 2n/3 \text{ or } X_r < n/3]$. $X_r$ is a binomial random variable in which the probability of success is $1/2$ (coin flip probability) and the number of trials is $n$. 

Let $F$ be the cumulative distribution function and $p_c = 1/2$ be the probability that a coin flip comes out to 1. For the upper tail, $Pr[X_r > 2n/3] < Pr[X_r \geq 2n/3] = F(n - 2n/3;n,1 - p_c) = F(n/3; n, 1/2)$. Similarly for the lower tail, $Pr[X_r < n/3] < Pr[X_r \leq n/3] = F(n/3; n, p_c) = F(n/3; n, 1/2)$. Therefore, $Pr[X_r > 2n/3 \text{ or } X_r < n/3] < 2 * F(n/3 ; n, p_c)$.

We will use the Hoeffding's inequality tail bound to upper bound $p$. 
\begin{center}
$F(k; n, p) \leq \exp(-2n(p - \frac k n)^2)$\\
$F(2n/3; n, 1/2) \leq \exp(-2n(1/2 - \frac {2n/3} n)^2)$\\
$= \exp(-2n(1/2 - 2/3)^2)$
$= \exp(-2n * 1/36)$
$= e^{-n/18}$
\end{center}

Then $Pr[X_r > 2n/3 \text{ or } X_r < n/3] < 2 * F(n/3 ; n, p_c) \leq 2 * e^{-n/18}$. If the probability that $Z_r$ equals 1 is $p$ for each round and we assume that $Z_r$ is independent for each $r$, then it must be that $E[R] = 1/p > 1/(2*e^{-n/18}) = e^{n/18}/2$. Therefore, $E[R]$ is exponential in $n$.
\end{proof}

\begin{theorem}
The adversary can sustain the delay attack for an expected exponential number of rounds.
\end{theorem}
\begin{proof}
As we described above and showed in Section~\ref{sec:delay}, the only way for the strategy to fail and a famousness decision to occur is if a coin round creates a supermajority. If no supermajority is created during a coin round, then the attack can carry on as normal until the next coin round.

We observed above that $E[R]$ corresponds to the expected number of coin rounds it takes for a supermajority to appear during our delay attack. Let $c > 1$ be the number of rounds between coin rounds. Then we expect it to take $c * e^{n/18}/2$ rounds for a supermajority to appear. Therefore, the adversary can sustain the attack for an expected exponential number of rounds.
\end{proof}

\end{document}